\documentclass[11pt,a4paper]{article}
\usepackage{cite}
\usepackage{url}
\usepackage{amssymb}
\usepackage{graphicx}
\usepackage{kbordermatrix}
\usepackage{multirow}
\usepackage{amsmath}
\usepackage{algorithm}
\usepackage[noend]{algorithmic}
\usepackage{caption}
\usepackage{color}
\usepackage{subfigure}
\usepackage{authblk}

\newtheorem{lemma}{Lemma}

\newtheorem{theorem}{Theorem}
\newtheorem{definition}{Definition}

\newtheorem{proof}{Proof}

\newcommand\Item[1][]{%
  \ifx\relax#1\relax  \item \else \item[#1] \fi
  \abovedisplayskip=0pt\abovedisplayshortskip=0pt~\vspace*{-\baselineskip}}

\begin{document}

\title{SMCL - Stochastic Model Checker for Learning in Games}

\author{Hongyang Qu, Michalis Smyrnakis and Sandor M. Veres}
\affil{Department of Automatic Control and Systems Engineering\\ University of Sheffield,
Sheffield 
S1 3JD, UK \\\{h.qu, m.smyrnakis, s.veres\}@sheffield.ac.uk}

\maketitle

\begin{abstract}
A stochastic model checker is presented for analysing the performance
of game-theoretic learning algorithms. The method enables the
comparison of short-term behaviour of learning algorithms intended for
practical use. The procedure of comparison is automated and it can be tuned for
accuracy and speed. Users can choose from among various learning
algorithms to select a suitable one for a given practical problem. The
powerful performance of the method is enabled by a novel
behaviour-similarity-relation, which compacts large state spaces into
small ones. The stochastic model checking tool is tested on a set of
examples classified into four categories to demonstrate the
effectiveness of selecting suitable algorithms for distributed
decision making.
\end{abstract}

\section{Introduction} \label{sec:intro}

Distributed decision making has recently received considerable
attention for the development of autonomous agent teams, and in
general, where agents need to make their decision in interaction with
their environment, including other agents. Where autonomy is a desired
property for a robotic agent application, for instance in remote or
dangerous fields, battle fields, disaster sites, mining and nuclear
waste processing applications, in such scenarios game-theoretic
learning algorithms can be used as coordination mechanism between
agents. Other applications are field robotics \cite{ifac}, wireless
sensor networks \cite{monitoring,sn}, smart grids
\cite{smart_grid1,smart_grid2}, disaster management \cite{cj,stelios}
and scheduling tasks \cite{sc}. Game theoretic learning can also be
useful in complex environment such as autonomous urban transport,
warehouses, airports and construction sites, as direct communication
between agents could be blocked or broken and consequently the system
has some random, stochastic behaviour.

When autonomy is a desirable property of an agent team then
distributed decision making becomes an important
attribute. Coordination mechanisms are needed to enable the agents to
accomplish a task or achieve a goal. Many learning algorithms have
been proposed in the literature for coordination mechanisms, which
differ in communication requirements, computational cost and presence
or absence of convergence proofs. Overall, there has however been no
technique available to study collective behaviour when only a small
number of iterations are allowed to achieve results.  From an
application point of view it is important to learn and accomplish some
desired task quickly due to various constraints.  This paper provides
methods to study learning performance during short periods.

Apart from short time periods for learning, a second hindrance of
effective cooperation can be the communication bandwidth and the
amount of communication needed for collective achievement by an agent
team.  A game-theoretic formulation, where agents mostly observe each
other, and communicate little, is favourable for many applications.
The game theoretic cooperation we consider in this paper needs very
little or no communication at all between agents. Lack of
communication is compensated by the agent's observing each other and
learning each others strategies.
 
A cooperative game-theoretic learning algorithm can be used to solve a
single decision problem. In a complex task, where a sequence of
decisions are needed, a sequence of algorithms can  be applied. 
If this approach is applicable, then games in strategic form are
generated for each type of decision making problem in a scenario of players, physical environment and rewards. 
This strategic form game can be seen as a
snapshot of a general learning framework, which models the whole
process of accomplishing the task. For example, partially observable
games (POGs) are the game theoretic equivalent of decentralised
Partially Observable Markov Decision Processes
(dec-POMDPs)~\cite{nair,emery2004,oliehoek}. A technique for solving a POG
is to generate strategic form games each time there is a change in the
scenario of the word., i.e. a strategic form game is based on 
the scenario of the world, meaning the players, the physical environment and the shared goal and reward functions.
In the snapshot, the scenario does not change any more, and players then choose actions to maximise
their rewards ~\cite{emery2004,emery2005} using game-theoretic learning algorithms. 
%% Therefore, the proposed verification
%% approach can find applications in tasks that involve POMDPs.

%In this paper is different from the parity games and thus synthesis which are well studied in verification literature. In contrast to parity games we do not formulate the decision optimisation task as an asynchronous decision making process between the environment and a system. But we allow individual agents/agents to choose actions simultaneously in order to achieve the team's goal.

Traditional approaches to analysing the performance of a
game-theoretic learning algorithm are usually based on convergence
proofs (Nash and Pareto Nash optima) or based on extensive simulation
studies. In simulation it is not always feasible or easy to understand
the agents' behaviour. The amount of simulations needed can be huge,
which can leave the possibilities of bad behaviour unexplored. This
paper introduces and analyses a stochastic model checking method which
bundles sets of state evolutions together to reduce complexity and
hence makes probabilistic verification feasible and practicable to
replace simulation studies.

% In addition they require a large number of simulation runs
% and even so, it is hard to establish the confidence about the
% performance because no rules have been set up to compute if the number
% of runs are sufficient.
% %
% Some work has been conducted to investigate the number of runs needed
% for simulation in order to study properties of stochastic models,
% e.g.,~\cite{Byrne13}. However, their results cannot be applied to our
% setting because the source of randomness in game-theoretic learning algorithms
% is different from that considered in other work. 
% %
% Model checking, as presented here, is superior
% to simulation in this case because of its exhaustive search. 
% To our
% best knowledge, our model checker is the first attempt
% to analyse the short term behaviour of various learning
% algorithms.  
%
%% {\it Probably there was some publication on performance or verification POMDPs, or similar results. We need to find these and refer to them and explain the fundamental difference with these in this paper. Also references to generic and most related model checking techniques should be referenced. Without placing this paper in the context of related international research our paper has no chance of credibility of containing important results! ! }
%
%To our best knowledge, our model checker MCGTL, {\em Model Checker for Game-Theoretic Learning algorithms}, is the first attempt to analyse the short term behaviour of various learning algorithms.  
The techniques presented explore the whole state space of a given learning problem. A state implies  a probability distribution of joint actions by game players in a Discrete-Time
Markov Chain (DTMC); a transition between two states represents a joint action associated with a firing probability. 
%Thus, the states and the transitions form a Discrete-Time Markov Chain (DTMC). 
Such a DTMC-based state-space model can contain a large number
of states due to the parameters used in some algorithms. We introduce a
{\em behaviour similarity relation} ``$\preceq$'' among states,
 %as defined in this paper, MCGTL generates a compact representation of the state space, 
which enables  efficient computation of steady state
probabilities using probabilistic model checking~\cite{KNP04a} for {\em
  Nash Equilibria}. The performance of a learning algorithm can be
measured by the probability at which a Nash
equilibrium will be observed after a finite number of iterations. Nash equilibria are good measures of performance as they can
be seen as  optimal solutions to distributed decision making
problems.
%% The reason that we focus on the probability a Nash
%% equilibrium can be observed after a finite number of iterations as
%% performance of a learning algorithm, is because a Nash equilibrium can
%% be seen as an optimal solution to the distributed decision making
%% problem. 
%% Nonetheless it is not the only performance measure and not
%% all applications of the proposed tool need to take this probability
%% into account.
However, our verification tool is not restricted to the cases
where steady states exist. It is possible that a learning algorithm
will not reach a steady state after a small number of iterations, or
that a specific sequence of joint actions will lead to a steady state
after a few iterations, while other sequences will not. The 
tool can be used to compute all these probabilities and analyse all
the possible short term outcomes of the learning-algorithms within a
given time period.
%% %
%% In the literature, there are probabilistic model
%% checkers, such as  PRISM~\cite{KNP11}, that implement the technique
%% in~\cite{KNP04a} for computing steady-state probabilities. But it
%% cannot generate the DTMC models for game-theoretic learning
%% algorithms.
%% %
%% Statistical model
%% checking (SMC)~\cite{LegayDB10} also use simulation runs to verify probabilistic
%% properties in DTMCs, including steady-state
%% properties~\cite{Rohr13,RabihGPV11}, but it is not clear to us how to
%% integrate the behaviour similarity relation in to SMC. Without this
%% relation, SMC would encounter the same problem on the number of
%% necessary runs as that in simulation.
%% %

% Some evaluation for an early prototype of our verification tool has
%% been published in~\cite{technical-report}. In this paper, we
%% integrated more DCOP algorithms into MCGTL, and evaluate the
%% algorithms with more DCOP problems in order to understand their
%% performance better. In addition, we improved the behaviour similarity
%% relation for the Particle Filter Fictitious Play algorithm, which
%% accelerate the verification process to a great extent.
%% %

The structure of the paper is as follows.  First our game theoretic framework is defined with learning, which is followed in
 Section \ref{sec:algo}  by simulation examples using Fictitious Play~\cite{brown_fict} on the coordination
game of Equation (\ref{toy_reward}) and the Shapley's game
\cite{shapley}.   
%Section~\ref{sec:prelim} presents preliminaries for the game-theoretical setting, a brief introduction to the learning
%algorithms analysable by our stochastic model checker. 
The approach for computing finite probabilities for a game with initial conditions
is presented in Section~\ref{sec:framework}. We show comparison among
all supported algorithms on several examples in Section~\ref{sec:exp},
and conclude the paper in Section~\ref{sec:concl}.

\section{Preliminaries} \label{sec:prelim}
\subsection{Game-theoretic definitions} \label{sec:games}
%In this section we will present a brief description of some basic game-theoretic definitions used in the rest of this paper. %The descriptions of the algorithms that we will use  as well.
%\subsection{Game-theoretic definitions}
%Many agent based problems can be cast as games.% , where agents are the players of the game and their actions are the available actions of the players in a game. 
%
\begin{definition}
A game in strategic form has the following elements \cite{games1}:%% : (1)
%% A set of players, $i={1,2,\ldots,\mathcal{I}}$, (2) A set of actions
%% for each player $i$, $A^{i}$ from which it selects its action $a^{i}$,
%% (3) A set of joint actions $a=(a^{1},a^{2},\ldots,a^{\mathcal{I}})$,
%% $a \in A$, $A=\times_{i=1}^{i=\mathcal{I}}A^{i}$, and (4) A reward function $r^{i}(a):A \rightarrow\mathbb{R}$
\begin{itemize}
\item A set of players, $i={1,2,\ldots,\mathcal{I}}$;
\item A set $A^{i}$ of actions for each player $i$, from which it selects its action $a^{i}$;
\item A set $A=\times_{i=1}^{i=\mathcal{I}}A^{i}$ of joint actions,
  where a joint action $a\in A$ is of the form $a=(a^{1},a^{2},\ldots,a^{\mathcal{I}})$; 
\item A reward function $r^{i}:A \rightarrow\mathbb{R}$, where
  $r^{i}(a)$ is the reward that player $i$ will gain if the joint action $a$ is played.
\end{itemize} 
\end{definition}
Given a joint action $a=(a^{1},a^{2},\ldots,a^{\mathcal{I}})$, we
often write $a^i\in a$ to indicate that $a^i$ ($1\le i\le
\mathcal{I}$) is a component of $a$.
We also refer to the set of joint actions of all players but $i$
as $A^{-i}$, and $r^{i}(a^{i},a^{-i})$ the reward that player $i$ will
gain if he plays action $a^i$ and its opponents play
$a^{-i}\in A^{-i}$. Often we write $r^{i}(a^{i})$, instead of
$r^{i}(a^{i},a^{-i})$, when $a^{-i}$ is of no interest. Let $\Delta^{i}$ be the set of all the probability distributions over
player $i$'s action space. A strategy $\sigma^{i}\in \Delta^{i}$ denotes the probability
distribution player $i$ uses to choose from its available
actions, and $\sigma^{i}(a^i)$ the probability of choosing action $a^i\in A^i$. A {\em pure} strategy is the case where
player $i$ choose a single action with probability 1. Other cases are {\em mixed} strategies. 
% Since the strategy of a player $i$ is a vector with length
% the cardinality of $A^{i}$ when we refer to the probability that
% he has to play action $a^{i}$, we will denote this as
% $\sigma(a^{i})$.
A joint strategy is defined as
$\sigma=\prod_{i=1}\limits^{\mathcal{I}}\sigma^i 
%%%%%(\sigma^{1}\times\sigma^{2}\times\ldots\times\sigma^{\mathcal{I}}) 
\in \times_{i=1}^{i=\mathcal{I}}\Delta^{i}$.
We often write $\sigma^{ -i}$ for a joint strategy of player $i$'s
opponents. 
We write $\sigma(a^{i},\sigma^{-i})$ for the probability of player $i$
playing action $a^{i}$ when its opponents use strategies $\sigma^{ -i}$.
% We will write $\sigma^{-i}$ for player $i$'s opponents' strategies,
% and $r^{i}(a^{i},\sigma^{-i})$ the expected reward
% player $i$ will gain if he chooses an action $a^{i}$ when his
% opponents choose the joint strategy $\sigma^{-i}$, i.e., $ r^{i}(a^{i},\sigma^{-i}) = E(r^{i}(a)| a^{-i} \sim \sigma^{-i} )$.
% We will write $s^{-i}$ and
% $\sigma^{-i}$ for player $i$'s opponents' actions and strategies
% respectively.
 By some abuse of notation, we use
$r^{i}(a^i,\sigma^{-i})$ to denote the expected reward of player $i$ when it
chooses action $a^i$ and its opponents use the joint strategy
$\sigma^{-i}$, i.e., 
$$r^i(a^i,\sigma^{-i})=\sum_{a^{-i} \in A^{-i}} r^i(a^i,a^{-i})\sigma^{-i}(a^{-i}),$$ 
where $\sigma^{-i}(a^{-i})$ is the
probability of playing the joint action $a^{-i}$ by player $i$'s
opponents. Similarly, for a strategy $\sigma^{i}$, we can write 
$$r^i(\sigma^i,\sigma^{-i})=\sum_{a^{i} \in A^{i}}\sum_{a^{-i} \in A^{-i}}
\sigma^{i}(a^{i})r^i(a^i,a^{-i})\sigma^{-i}(a^{-i}).$$

To choose actions, players can use either deterministic or stochastic
rules based on joint strategies. A common deterministic rule is {\em
  best response} (BR) $\hat{\sigma}$ in Equation~(\ref{eq:BR}) where
players choose an action maximising their expected reward. Therefore
the payers are rational, as they always increase or maintain their
reward, given that the other players are rational players too and the
estimates of the other players' strategies are correct.

\begin{equation}
\hat{\sigma}^{i}= \mathop{\rm argmax}_{\sigma^{i} \in \Delta^{i}} \quad
r^{i}(\sigma^{i},\sigma^{-i})
\label{eq:BR}
\end{equation}
\noindent {\bf Remark.} Let us introduce some arbitrary mapping of
$\mathcal A^{i}$ into $M:=\{1,\ldots,|A^{i}|\}$ where $|A^{i}|$, is
the cardinality of $A^{i}$. Denote a generic element of $M$ by $j$. In
other words, the $j$ is an indexing of the elements of $M$
according to some arbitrary but fixed ordering. If a
game has more than one best response, then we consider that players
will always choose the action with the smallest index $j$.

{\em Smooth best response}
$\bar{\sigma}$ in Equation~(\ref{eq:smbr}) is the most common stochastic rule.
\begin{equation}
\bar{\sigma}^{i}(a^{i},\sigma^{-i})=\frac{\exp(r^{i}(a^{i},\sigma^{-i})/\tau)}{\sum_{\tilde{a}^{i}\in
  A^i}\exp(r^{i}(\tilde{a}^{i},\sigma^{-i})/\tau)}
\label{eq:smbr}
\end{equation}
where $\tau$ is a randomisation parameter.
% %\noindent 
% \begin{minipage}[b]{0.4\textwidth}
% %\centering
% \begin{equation}
% \hat{\sigma}^{i}= \mathop{\rm argmax}_{\sigma^{i} \in \Delta^{i}} \quad
% r^{i}(\sigma^{i},\sigma^{-i})
% \label{eq:BR}
% \end{equation}
% \end{minipage}
% \begin{minipage}[b]{0.6\textwidth}
% %\centering
% \begin{equation}
% \bar{\sigma}^{i}(a^{i},\sigma^{-i})=\frac{\exp(r^{i}(a^{i},\sigma^{-i})/\tau)}{\sum_{\tilde{a}^{i}\in
%   A^i}\exp(r^{i}(\tilde{a}^{i},\sigma^{-i})/\tau)}
% \label{eq:smbr}
% \end{equation}
% \end{minipage}

With a relatively large $\tau$, e.g., $\tau=0.01$, smooth best
response allows players to choose actions that do not maximise their
expected reward with non-zero probability. When
players use a small $\tau$, e.g., $0.0001$, they intend to
choose with probability 1 the action that maximises their expected
reward, which is equivalent to best response. A small $\tau$ is
preferred in many practical applications because of its deterministic
nature, which makes the game playing predictable. On the other hand,
it reduces the chances of exploring different actions, which
could lead to find better solutions either in few iterations or with  high probability \cite{May}. In the rest
of the paper, we combine two decision rules together by using
smooth best response with $\tau=0.01$ in the first iteration and best response in other
iterations. This is a very simple annealing scheme~\cite{gwfp} to make
a balance between probabilistic and deterministic behaviour.

Nash in \cite{Nash} proved that every game has at least one
equilibrium which is a fixed point in the best response
correspondence. Thus, when players choose their actions using Equation
(\ref{eq:BR}), a strategy $\hat{\sigma}$ is a Nash equilibrium
if
\begin{equation}
r^{i}(\hat{\sigma}^{i},\hat{\sigma}^{-i})\geq r^{i}(\sigma^{i},\hat{\sigma}^{-i}) \qquad \textrm{for all }  \sigma^{i} \in \Delta^{i}, i=1,\ldots,\mathcal{I}.
\label{eq:nashutil}
\end{equation}
Equation (\ref{eq:nashutil}) implies that if the joint strategy
$,\hat{\sigma}$ is played in a game and it is a Nash equilibrium, then no
player can increase its rewards by unilaterally changing its own
strategy.
When all the players in a game select their actions using pure (mixed, resp.)
strategies, the equilibria are called {\em pure (mixed, resp.) strategy
Nash equilibria}. An equilibrium is called
{\em Pareto efficient} if no other joint mixed strategy can increase the
payoff of all the players. 

%
%\textcolor{red}{Explain from a Nash equilibrium, players cannot
%  execute probabilistic strategies to go outside the equilibrium.}
%
\noindent
{\bf Example.} Equation~(\ref{toy_reward}) shows a simple coordination game
\begin{equation} \label{toy_reward}
r=\kbordermatrix{ &a_1& &a_2\\
b_1&1,1&\vrule&0,0\\
b_2&0,0&\vrule&1,1
}
\end{equation}
In this example, player 1 can choose action $b_1$ or $b_2$, and player
2 can choose $a_1$ or $a_2$. In each row and column intersection the rewards of players 1 and 2 are listed. The maximum reward, and pure Nash
equilibrium of the game, is 1 for both players when they execute the
joint action $(b_1, a_1)$ or $(b_2,
a_2)$. 

\subsection{Basic game-theoretic learning algorithms} \label{sec:algo} 

% Distributed optimisation is a crucial component of many applications
% such as sensor networks \cite{sn} and disaster management
% \cite{cj,stelios}.
% In many applications, agents coordinate to complete their tasks by
% maximising a common reward. For that reason, 
Iterative learning
algorithms have been introduced as negotiation mechanisms between the
agents. Typically, such an algorithm works as follows \cite{learning_in_games}. Each player
repeatedly plays the game by trying to estimate other players' next move based on the observation of their behaviour history, and then
selecting its own action to maximise its reward. They stop the
iterative assessment if either they converge to an equilibrium
 or the maximum number of iterations is reached.

\noindent 
{\bf Fictitious Play}~\cite{brown_fict} is the canonical example of
game-theoretic learning. Each player $i$ in the initial iteration has
a non-negative weight function $\kappa^{i\rightarrow j}_{0}\in (0,1)$
for the actions of each of its opponents $j$ ($j\in\{1,\ldots,
\mathcal{I}\} \backslash \{i\}$) and update this weight function after
observing player $j$'s action. Note here that the initial weights can
be any arbitrary positive real number but this can significantly influence the
estimations of opponents strategies and for that reason we choose
to normalise the initial weights in order to be between zero and
one. Therefore, we have 
\begin{equation}
\sum_{a^j \in A^j} \kappa_{0}^{i\rightarrow j} (a^j)= 1. \label{eq:kappa0}
\end{equation}

Based on these weights, each player estimates its opponents' strategies and choose an
action that maximises its expected reward.
Let $\kappa_{t}^{i\rightarrow j}(a^j)$ be player $i$'s weight for player
$j$'s action $a^j\in A^j$ at the $t$-th iteration, and $a_t^j\in A^j$ the action chosen by player $j$ at the $t$-th iteration. %% During the
%% initialisation process, i.e., $t=0$,
%% $\kappa_{0}^{i\rightarrow j}(a^j)$ is randomly chosen.
At the $t$-th
iteration, the weight is updated as follows.

%% \noindent 
%% \begin{minipage}[b]{0.5\textwidth}
%% \centering
%% \begin{equation}
%% \kappa_{t}^{i\rightarrow j}(a^{j}) = \kappa_{t-1}^{i\rightarrow j}(a^{j})+\mathbb{I}_{a^j_t=a^j},
%% \label{eq:kappa}
%% \end{equation}
%% \end{minipage}
%% \begin{minipage}[b]{0.5\textwidth}
%% \centering
%% where
%% $\mathbb{I}_{a^j_t=a^j}=\left\{\begin{array}{cl}1&\mbox{if
%% $a^j_t=a^j$}\\0&\mbox{otherwise.}\end{array}\right.$ 
%% \end{minipage}
\begin{equation}
\kappa_{t}^{i\rightarrow j}(a^{j}) = \kappa_{t-1}^{i\rightarrow j}(a^{j})+\mathbb{I}_{a^j_{t-1}=a^j},
\label{eq:kappa}
\end{equation}
where
$\mathbb{I}_{a^j_t=a^j}=\left\{\begin{array}{cl}1&\mbox{if
$a^j_{t-1}=a^j$}\\0&\mbox{otherwise.}\end{array}\right.$ 

Equation~(\ref{eq:kappa}) simply increases the weight by 1 for the
action executed in the previous iteration.
Fictitious play is based on the implicit assumption that players use
the same strategy in all iterations of the game. Based on this
assumption, player $i$ uses a multinomial distribution to estimate
player $j$'s strategy $\sigma_{t}^{i\rightarrow j}(a^{j})$ by the
following formula:
\begin{equation}
\sigma_{t}^{i\rightarrow j}(a^{j})=\frac{\kappa^{i\rightarrow j}_{t}(a^j)}{\sum_{a' \in A^{j}}\kappa^{i\rightarrow j}_{t}(a')}.
\label{eq:fp1p}
\end{equation}
Player $i$ chooses the action which maximises its expected reward,
based on Equation (\ref{eq:fp1p}) and (\ref{eq:smbr}). 
For Equation (\ref{eq:smbr}),
$\sigma^{-i}=\times_{j=1}^{j=\mathcal{I}\land j\not= i}\sigma^{i\rightarrow j}.$

\noindent {\bf Remark.} We should mention here that $BR$ is a
deterministic decision making rule in the sense that players always
choose the action which maximises their expected reward. Nonetheless,
it can lead to players using mixed strategies when a number of
actions can be played interchangeably in the iterations of the
game. If this is the case, the probability of choosing each of such
actions is always greater than $0$.
 
FP can be seen as the following three step process: (1) observe the
opponents' actions; (2) update the beliefs about the opponents'
strategies; (3) choose an action based on (smooth) best response decision rule.
Equation (\ref{eq:fp1p}) can also be written as: 
%which can be also written as: 
\begin{equation} \label{eq:fpwitha}
\begin{split}
\sigma_{t}^{i\rightarrow j}(a^{j}) & = \frac{\kappa^{i\rightarrow
    j}_{t}(a^j)}{t+\sum_{a^{j}\in A^{j}}\kappa^{i\rightarrow
    j}_{1}(a^j)} \\
& = \frac{\kappa^{i\rightarrow
    j}_{t-1}(a^j) + \mathbb{I}_{a^j_{t-1}=a^j}}{t+1} \\
%\quad (\mbox{Equations~(\ref{eq:kappa0}) and (\ref{eq:kappa})})\\
& = \frac{t \cdot \kappa^{i\rightarrow j}_{t-1}(a^j)}{t\cdot
  (t+1)}+\frac{\mathbb{I}_{a^j_{t-1}=a^j}}{t+1} \\
& = \frac{t}{t+1}\sigma_{t-1}^{i \rightarrow j}(a^{j})+\frac{\mathbb{I}_{a^j_{t-1}=a^j}}{t+1}\\
& = (1-\frac{1}{t+1})\sigma_{t-1}^{i \rightarrow j}(a^{j})+ \frac{1}{t+1}\mathbb{I}_{a^j_{t-1}=a^j}.
\end{split}
\end{equation} 

In addition to FP, our probabilistic model checker also supports the following learning
algorithms: 
%%%%%%%%%Joint strategy fictitious play (JSFP) \cite{marden_jsfp},
Geometric fictitious play (GFP) \cite{learning_in_games}, and Adaptive
forgetting factor fictitious play(AFFFP) \cite{mythesis}.
These algorithms can be classified into the same category: 
{\em Algorithms based on weighted average history of actions}. They take into
account all the history of players' actions in order to estimate their
opponents' strategies. Actions are then selected based on these
estimations. The differences of these algorithms are based on the
impact of the recently observed action on the estimation of the
other players' strategies.
%%%%%%%%%%%% JSFP aims at reducing the computational complexity
%% of FP for a large number of players and the basic difference with FP
%% is that each player updates directly its reward function, instead of
%% its weight functions. On the other hand, GFP and AFFFP 
They increase the
impact of the recent observation using discount factors when estimating $\sigma^{-i}$.%, instead of the average of the historical data.
% : myopic algorithms and algorithms that
% are based on an average or weighted average over the observed actions
% (algorithms that are based on fictitious play). 
% and algorithms were prediction of other players' strategies are used.

\noindent
{\bf Example.}
Given the initial estimation in Equation~(\ref{toy_weight}) for the game in (\ref{toy_reward}),
\begin{equation} \label{toy_weight}
\kappa_{0}^{1\rightarrow 2}=[0.511, 0.489]^{T} \mbox{ and } \kappa_{0}^{2 \rightarrow 1}=[0.489, 0.511],
\end{equation}
 we use smooth best response in Equation~(\ref{eq:smbr}) with
 $\tau=10^{-2}$ to compute the initial
 joint strategies, and use best response
 thereafter. When the game is started, we obtain that  
$\sigma_{0}^{1\rightarrow 2}=\sigma_0^{-1}=\kappa_{0}^{1\rightarrow 2}$, $\sigma_{0}^{2 \rightarrow 1}=\sigma_{0}^{-2}=\kappa_{0}^{2 \rightarrow 1}$. 
Using the above weights, equation (\ref{eq:smbr}) and $\tau=0.01$  we obtain $\bar\sigma_{0}^{1}=[0.9, 0.1]^{T}$ and $\bar\sigma_{0}^{2}=[0.1, 0.9]$. This means that player 1 has probability to choose actions 
$b_1$ and $b_2$ with probability $0.9$ and $0.1$
respectively, and player 2 chooses action $a_1$ and $a_2$ with probability $0.1$ and $0.9$
respectively. 
%% Equation~(\ref{init_dist}) shows the joint strategy $\sigma_0$. 
%% \begin{equation} \label{init_dist}
%% \sigma_0=\bar\sigma_0^1\times \bar\sigma_0^2 =\kbordermatrix{ &a_1&a_2\\
%% b_1&0.09&0.81\\
%% b_2&0.01&0.09
%% }
%% \end{equation}

%% \noindent
The initial joint strategy is shown in the ellipse node in Figure \ref{eq:dtmctree}.
It means that the
probability of executing $(b_1, a_1)$, $(b_1, a_2)$, $(b_2, a_1)$ and
$(b_2, a_2)$ is $0.09$, $0.81$, $0.01$ and $0.09$
respectively. If the players execute $(b_1, a_1)$, then they update
their weight function using Equation~(\ref{eq:kappa}) and obtain
\begin{equation} \label{succ_weight}
\kappa_{1}^{1\rightarrow 2}=[1.511, 0.489]^{T} \mbox{ and } \kappa_{1}^{2 \rightarrow 1}=[1.489, 0.511].
\end{equation}
The new joint strategy generated using Equations~(\ref{eq:fp1p}) and
(\ref{eq:BR}) is shown in the top left rectangular node in Figure \ref{eq:dtmctree}.
%% Equation~(\ref{succ_dist}).
%% \begin{equation} \label{succ_dist}
%% \sigma_1=\kbordermatrix{ &a_1&a_2\\
%% b_1&1&0\\
%% b_2&0&0
%% }
%% \end{equation}

%% An example of the possible simulation traces that can be generated by
%% the above described process is depicted in Figure
%% \ref{eq:dtmctree}.

\begin{figure}
\centering
  \includegraphics[scale=0.43]{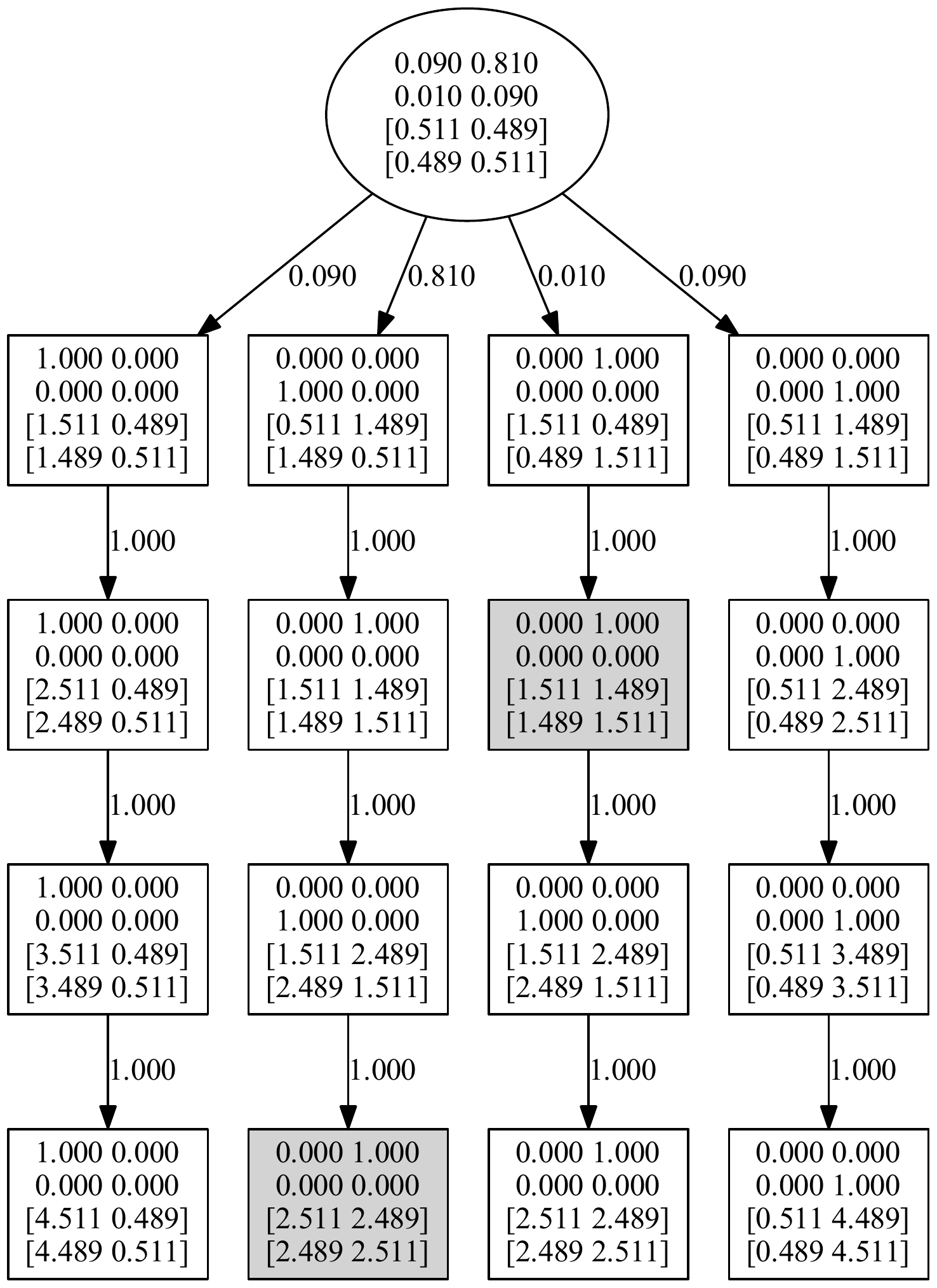}
  \captionof{figure}{Five layer expansions of possible simulation traces.}
  \label{eq:dtmctree}
\end{figure}

\begin{figure}
  \centering
    \includegraphics[scale=0.44]{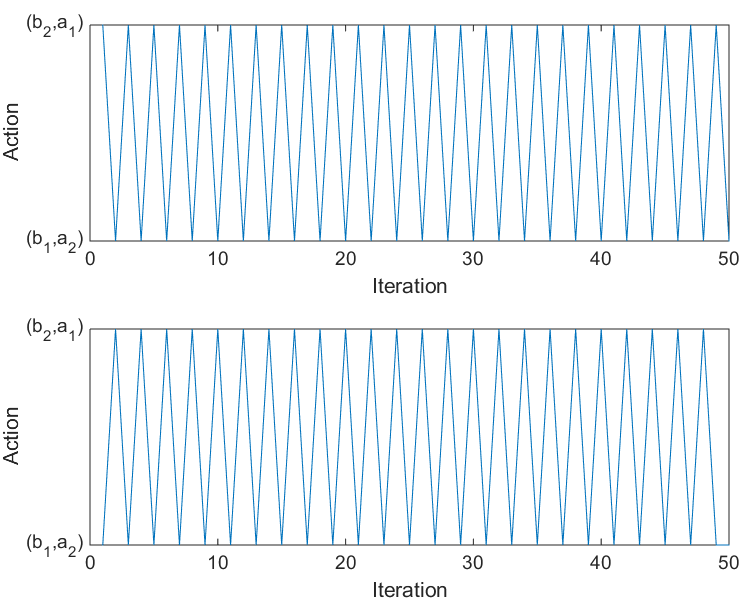}
  \captionof{figure}{Simulations for the game in Equation (\ref{toy_reward})} \label{fig:ex2}
\end{figure}

Figure \ref{fig:ex2} are examples showing that it
is difficult or impossible to understand agents' behaviour from
simulation. Figure \ref{fig:ex2} shows that FP
being trapped in a cycle when the game is repeatedly played for the
initial conditions in Figure \ref{eq:dtmctree}. The top figure
illustrates the selected actions when the branch with initial
probability 0.81 is expanded and the bottom figure the selected
actions when the branch with initial probability 0.01 is expanded. In
this simple case, we can observe that the two simulation runs are
similar. The same cycle of actions is repeated in both figures, but
they have one iteration lag. Thus, in the top figure joint action
$(b_{1},a_{2})$ is played in the odd iterations, while in the bottom
figure it is played in the even iterations. In more complicated cases,
where joint actions are repeated with lag greater than one, or the lag
is unknown, it would be hard or even infeasible to determine if two
simulation runs are similar.

\subsection{Other learning algorithms}
%%%%%%%%%%%%%%% {\bf Joint Strategy Fictitious Play} (JSFP) was introduced in
%% \cite{marden_jsfp} in order to reduce the computational complexity of
%% FP for large number of players. In the classic FP algorithm players
%% have to update weight functions and estimate expected rewards. In
%% order to reduced the computational cost of this process, an update rule that applies directly to the
%% estimation of the expected rewards was proposed:
%% \begin{equation}
%%  r^{i}_{t}(a^i)=\frac{t-1}{t}r^{i}_{t-1}(a^i)+\frac{1}{t}(r^{i}(a^{i},a^{-i}_{t-1})),
%% \label{eq:jsfp_util}
%% \end{equation}
%% where $r^{i}_{t}(a^i)$ is the reward that the action $a^i\in A^i$
%% could gain at the $i$-th iteration, and
%% $r^{i}(a^{i},a^{-i}_{t-1})$ is the reward that player $i$ will
%% gain if he plays action $a^i$ given that the joint action of his opponents
%% at the previous iteration is $a^{-i}_{t-1}$. The initial value of reward
%% that $a^i$ has is irrelevant because $\frac{t-1}{t} = 0$ when $t=1$.
%% Based on (\ref{eq:jsfp_util}) players can use either best or smooth best response in order
%% to choose actions.

{\bf Geometric Fictitious Play.}
In FP all the previous actions are of the same importance because of
the assumption that opponents' strategies are stationary. Geometric
fictitious play
(GFP) is a variant of FP that addresses this incorrect
assumption by exponential decaying the importance of the historical
data, by a factor $\alpha$ ($0<\alpha<1$). Thus, the recently observed
actions are more important in the estimation of the opponents'
strategies, which is based on the following formula:

%% \noindent 
%% \begin{minipage}[b]{0.55\textwidth}
%% \centering
%% \begin{equation}
%% \sigma_{t}^{i\rightarrow j}(a^{j})=(1-\alpha)\sigma_{t}^{i\rightarrow j}(a^{j})+ \alpha \mathbb{I}_{a^j_t=a^j},
%% \label{eq:gfp}
%% \end{equation}
%% \end{minipage}
%% \begin{minipage}[b]{0.45\textwidth}
%% \centering
%% where $\mathbb{I}_{a^j_t=a^j}=\left\{\begin{array}{cl}1&\mbox{if
%% $a^j_t=a^j$}\\0&\mbox{otherwise.}\end{array}\right.$
%% \end{minipage}
\begin{equation}
\sigma_{t}^{i\rightarrow j}(a^{j})=(1-\alpha)\sigma_{t}^{i\rightarrow j}(a^{j})+ \alpha \mathbb{I}_{a^j_{t-1}=a^j},
\label{eq:gfp}
\end{equation}
where $\mathbb{I}_{a^j_{t-1}=a^j}=\left\{\begin{array}{cl}1&\mbox{if
$a^j_{t-1}=a^j$}\\0&\mbox{otherwise.}\end{array}\right.$

{\bf Adaptive Forgetting Factor Fictitious Play.}
Even though GFP addresses the incorrect assumption of FP that the
strategies of other players are constant, it requires the players
to know in advance the rate the other players change their strategy
since it uses a constant $\alpha$ throughout the game. In order to
overcome this limitation, adaptive forgetting factor fictitious play
(AFFFP) was introduced in \cite{mythesis}. In this variant of FP there
is a time varying discount factor $\lambda_{t-1}$ which is adjusted
based on the likelihood of the previously observed actions. Opponents'
strategies are estimated as in Equation (\ref{eq:fp1p}),
%% \noindent 
%% \begin{minipage}[b]{0.36\textwidth}
%% \centering
%% \begin{equation}
%% \sigma_{t}^{i\rightarrow j}(a^{j})=\frac{\kappa_{t}^{i\rightarrow j}(a^{j})}{n_{t}^{j}}.
%% \label{eq:affp_sigma}
%% \end{equation}
%% \end{minipage}
%% \begin{minipage}[b]{0.64\textwidth}
%% \centering
%% where 
%% %\begin{displaymath}
%% $\kappa_{t}^{i\rightarrow j}(a^{j})=\lambda_{t-1}^{j} \kappa_{t-1}^{i\rightarrow j}(a^{j})+ \mathbb{I}_{a^{j}_{t-1}=a^{j}}.$
%% %\label{eq:afffp_kappa}
%% %\end{displaymath}
%% \end{minipage}
%% \begin{equation}
%% \sigma_{t}^{i\rightarrow j}(a^{j})=\frac{\kappa_{t}^{i\rightarrow j}(a^{j})}{n_{t}^{j}}.
%% \label{eq:affp_sigma}
%% \end{equation}
but the weights of opponents strategies are discounted as follows:
\begin{equation}
\kappa_{t}^{i\rightarrow j}(a^{j})=\lambda_{t-1}^{j} \kappa_{t-1}^{i\rightarrow j}(a^{j})+ \mathbb{I}_{a^{j}_{t-1}=a^{j}}
\label{eq:afffp_kappa}
\end{equation} 
where $\mathbb{I}_{a^{j}_{t-1}=a^{j}}$ is the same identity function
as in Equation (\ref{eq:kappa}). 
%% Similarly to the classic fictitious
%% play process the normalisation constant of $\kappa_{t}^{i\rightarrow
%%   j}$, $n_{t}^{j}$, is defined as 
Let $n_{t}^{j}=\sum_{a^{j}\in A^{j}}\kappa_{t}^{i\rightarrow
  j}(a^{j})$ be the normalisation divisor in Equation (\ref{eq:fp1p}). Based on (\ref{eq:afffp_kappa}), we can use the following recursion to evaluate $n^{j}_{t}$:
\begin{displaymath}
n_{t}^{j}=\lambda_{t-1}^{j}n_{t-1}^{j}+1, 
\label{eq:affp_ni}
\end{displaymath} 
where
\begin{displaymath}
\lambda_{t}=\lambda_{t-1}+ \gamma \Big( \frac{1}{\kappa_{t-1}^{i \rightarrow j}(a)}\frac{\partial}{\partial \lambda} \kappa_{t-1}^{i \rightarrow j}(a) - \frac{1}{n^{j}_{t-1}}\frac{\partial}{\partial \lambda} n^{j}_{t-1} \Big),
\label{eq:affp_update}
\end{displaymath}
where $\gamma\in (0, 1]$ is a learning rate parameter which controls the impact of the observed actions in the new value of the adaptive factor $\lambda$, 
\begin{equation*}
\frac{\partial}{\partial \lambda}
  \kappa_{t}^{i\rightarrow j}(a) \vert_{\lambda=\lambda_{t-1}}= \kappa_{t-1}^{i \rightarrow j}(a)+
  \lambda_{t-1}\frac{\partial}{\partial \lambda} \kappa_{t-1}^{i \rightarrow j}(a)
  \vert_{\lambda=\lambda_{t-1}}	
\end{equation*}
 and 
\begin{equation*}
\frac{\partial}{\partial
    \lambda} n^{j}_{t}\vert_{\lambda=\lambda_{t-1}}=n^{j}_{t-1}+
  \lambda_{t-1}\frac{\partial}{\partial \lambda}
  n^{j}_{t-1}\vert_{\lambda=\lambda_{t-1}}. 
\end{equation*}

\subsection{Discrete-Time Markov Chains} \label{sec:dtmc}

%% DTMC is a mathematical formalism to model the state
%% space generated by our model checker. 
Given a set $S$, let
$Dist(S)$ be the set of all discrete probability distributions over
$S$, each of which is of the form $p: S\rightarrow [0, 1]$ such that
$\sum_{s\in S} p(s)=1$.

\begin{definition}
A DTMC $M$ is a tuple $\langle S, \dot{s}, T, \mathcal{P}, L \rangle$, where
\begin{itemize}
\item $S$ is a set of {\em states},
\item $\dot{s}\in S$ is the {\em initial state},
\item $T: S \rightarrow Dist(S)$ is a probabilistic transition
  relation among states. We often write $T_s(s')$ as the probability
  of moving from state $s$ to state $s'$,
\item $\mathcal{P}$ is a set of {\em atomic propositions}, or simply proposition.
\item $L: \mathcal{P}\rightarrow 2^{S}$ is a labelling function that
  maps a proposition to a subset of states. A
  proposition $\tt{p}\in \mathcal{P}$ {\em holds} in state $s\in S$ iff $s\in L({\tt{p}})$.
\end{itemize}
\end{definition}

A {\em finite} path $\rho$ is a sequence of states $s_0s_1,\ldots,s_n$ such
that $T_{s_i}(s_{i+1})>0$ for any $0\le
i<n$. The path $\rho$ is {\em infinite} when $n \to \infty$.
Various properties in DTMCs can be verified by probabilistic model
checking. In this paper, we are interested in {\em steady state}
properties, which specify the behaviour of a DTMC in the long
run, and are captured by {\em bottom strongly connected components} (BSCC).

\begin{definition}
A set of states $\overline{S}\subseteq S$ in a DTMC $M$ is a strongly
connected component (SCC)
iff the following condition holds:
\begin{eqnarray}
\forall s, s'\in \overline{S}: \mbox{ there
  exists a path } \rho=s_0s_1,\ldots,s_n \nonumber \\
  \quad \mbox{ such that } s=s_0
\mbox{ and } s'=s_n.
\end{eqnarray}
\end{definition}

\begin{definition}
A set of states $\overline{S}\subseteq S$ in a DTMC $M$ is a BSCC
iff it is an SCC with the following constraint:
\begin{equation}
\forall s\in \overline{S} \mbox{ and } s'\in S: T_s(s') > 0 \mbox{
  implies that } s'\in \overline{S}.
\end{equation}
\end{definition}
Intuitively, a BSCC, $bscc$ hereafter, is a special strongly connected component
of $M$ such that no states outside $bscc$ can be reached from
any states inside $bscc$ via the probabilistic transition relation. It
means that once the system enters $bscc$, it is trapped in it and
cannot escape. BSCCs in a DTMC can be identified by applying the
Tarjan's algorithm~\cite{tarjan72} to partition the DTMC into SCCs first and
then looking for SCCs that do not have outgoing transitions. 
%% A steady
%% state property can be written as $\cal{S}=? [\tt{p}]$ in PCTL, a
%% specification language in PRISM~\cite{KNP11}. This property asks what
%% the probability of $\tt{p}\in \mathcal{P}$ being satisfied in the long
%% run is.  
The model checking algorithms for computing steady-state probabilities
from BSCCs can be found in~\cite{KNP04a}.

\begin{theorem}
A Nash equilibrium is a BSCC. % of a DTMC.
\end{theorem}
\begin{proof}
Consider a state $s$ that represents the joint strategy $\sigma$
which is Nash equilibrium. If we assume that $s$ is not in a BSCC, then a
new state $\tilde{s}$ will be reached with a non-zero probability that
will represent a new joint strategy $\tilde{\sigma}$. This is a
contradiction because a joint strategy $\sigma$ is a Nash
equilibrium iff no player will deviate from $\sigma$.  
\end{proof}

%\begin{prop}
% \noindent
% {\bf Remark.} 
Note here that a BSCC is not necessarily a Nash equilibrium.
%% \label{prop2}
%% \end{prop}
%% \begin{proof}
%% In order to show that Proposition \ref{prop2} is correct, it is enough
%% to show that there exists a BSCC of a DTMC that it is not a
%% Nash equilibrium. The learning algorithms can either be trapped in
%% cycles or fail to converge to a Nash equilibrium.
A well known example
\cite{learning_in_games} comes from the FP algorithm, when it is used in the game 
in Equation~(\ref{toy_reward}). It has been shown that FP can
fail to converge in the Nash equilibrium in this
game and be trapped in a cycle between the actions $(b_{2}, a_{1})$
and $(b_{1}, a_{2})$, which is also depicted in Figure \ref{fig:fp}. This is 
a BSCC, but not a Nash equilibrium. 
%Therefore it is possible to observe a BSCC of the DTMC that it is not a Nash equilibrium. 
%\end{proof}
%
For a Pareto efficient equilibrium, all states in the BSCC have the
same reward. However, it is not true for a non-Pareto efficient
one. In this case, we are interested in the probability by which the
system stays in some of the states in the BSCC. In the example defined
in Equation~(\ref{toy_reward}), the Pareto efficient equilibria are captured by
single-state BSCCs that fire $(b_1, a_1)$ or $(b_2, a_2)$ only, i.e.,
with probability 1.

\section{Analysis of game-theoretic learning algorithms} \label{sec:framework}

We aim to build a framework for 
%overcome the drawback of simulation-based 
performance
evaluation of game-theoretic learning algorithms. It explores all
states an algorithm can visit.  Each state contains a joint strategy
$\sigma$ generated from the Cartesian product of the individual
strategy of all players,  
%The initial estimates can either given manually or generated
%randomly.  
and extra
information $\pi$, which varies among the learning algorithms.

\subsection{State space generation} \label{subsec:state-gen}

\begin{definition} 
A state $s$ is a tuple
$\langle\sigma,\mathbb{\pi} \rangle$, where
\begin{itemize}
\item
$\sigma$ is the joint strategy of all players, which indicates the likelihood of firing a joint action,
\item
$\mathbb{\pi}$ is a vector of parameters, including weights and other
  input variables for the learning algorithm.
\end{itemize}
Let $\kappa^i=\{\kappa^{i\rightarrow j}\mid j\in\{1,\ldots, \mathcal{I}\} \backslash \{i\}\}$ be the weights that player $i$ holds  about all its
opponents.
The vector $\mathbb \pi$ for each learning algorithm is specified as
follows:%\footnote{The details about these algorithms can be found in Appendix.}:
\begin{itemize}
\item FP: $\mathbb \pi =\langle \kappa^{1},\ldots,\kappa^{\mathcal{I}} \rangle$, 
\item GFP: $\mathbb \pi =\langle \kappa^{1},\ldots,\kappa^{\mathcal{I}}, \alpha \rangle$, 
\item AFFFP: $\mathbb \pi =\langle \kappa^{1},\ldots,\kappa^{\mathcal{I}}, n_{1},\ldots,n_{\mathcal{I}}, \lambda,\frac{\partial}{\partial \lambda} \kappa^{1}, \ldots,\frac{\partial}{\partial \lambda} \kappa^{\mathcal{I}}$, $\frac{\partial}{\partial \lambda} n_{1},\ldots,\frac{\partial}{\partial \lambda} n_{\mathcal{I}} \rangle$,
%\item EKFFP: $\mathbb{\pi}\!=\!\langle m^{-}_{1}\!\!,\!\ldots\!, m_{\mathcal{I}}^{-}\!\!,m_{1},\!\ldots\!,m_{\mathcal{I}},P^{-}_{1},\!\ldots\!,P^{-}_{i}\!,P_{1},\!\ldots\!,P_{\mathcal{I}}\!\rangle$,
%\item PFFP: $\mathbb{\pi}=\langle
 % Q_{1},\ldots,Q_{\mathcal{I}},w_{1},\ldots,w_{\mathcal{I}}\rangle$,
\end{itemize}
\end{definition}

Given a game and an algorithm, our probabilistic model checker adopts Breadth-First Search (BFS)
to explore the state space from the initial state $\dot{s} =
\langle\dot{\sigma}, \dot{\pi}\rangle$ by firing each joint action
with non-zero firing probability and computing a new state. The
exploration algorithm is shown in
Algorithm~\ref{algo2}. 
% It can be adopted for other learning algorithms by using $S$, instead of
% $S\setminus\{\dot{s}\}$ in Line 8, because those algorithms do not use
% weight functions.
As we are interested in
short-term behaviour, an upper bound on the depth in BFS can be set to
terminate the exploration if it does not stop within the bound. If
this is the case, then all unexplored states have a transition to a
special termination state $\bot$ with probability 1.

\begin{algorithm} \caption{State space exploration} 
\label{algo2}
\begin{algorithmic}[1]
\STATE{$s_0 := \langle\dot{\sigma}, \dot{\pi}\rangle$;
  $S:=\{s_0\}$; $\mbox{enqueue}(queue1,s_0)$}
\STATE{$depth:=0$; $i:=1$}
\WHILE{$queue1 \not= \emptyset$}
\STATE{$(s=\langle\sigma, \pi \rangle):=\mbox{dequeue}(queue1)$}
\FORALL{$a\in A$}
  \IF{$\sigma(a) > 0$}
    \STATE{$\langle \sigma'', \pi' \rangle :=
      \mbox{UpdateEstimations}(\pi, a)$}
    \STATE{$\sigma' := \mbox{UpdateProbDist}(u, \sigma'', \tau)$}
    \STATE{$s_i := \langle\sigma', \pi'\rangle$; $found := false$}
    \FORALL{$j= i-1,\cdots, 1$}
      \IF{$s_j\preceq s_i$}
        \STATE{$T_s(s_j) := \sigma(a)$; $found := true$; {\bf break}}
      \ENDIF
    \ENDFOR
    \IF{$found = false$}
      \STATE{$T_s(s_i) := \sigma(a)$; $S:=S\cup \{s_i\}$}
      \STATE{$\mbox{enqueue}(queue2,s_i)$; $i:=i+1$}
    \ENDIF
  \ENDIF
\ENDFOR
\STATE{$queue1:=queue2$; $queue2:=\emptyset$; $depth:=depth+1$}
\IF{$depth\ge {\tt UpperBound}$}
\STATE{$S:=S\cup \{\bot\}$; $T_{\bot}(\bot) := 1$}
\FORALL{$s\in queue1$}
  \STATE{$T_s(\bot) := 1$}
\ENDFOR
\STATE{{\bf break}}
\ENDIF
\ENDWHILE
\RETURN $(S,~T)$
\end{algorithmic}
\end{algorithm}

In Algorithm~\ref{algo2}, $queue1$ is used to store
states that need to be expanded, and $queue2$ stores their success
states. This is an easy way to count the depth of the
exploration. $\mbox{Enqueue} (queue, s)$ appends $s$ to the tail of
$queue$ and $s:=\mbox{dequeue}(queue)$ removes the head of $queue$ and
save it to $s$.
$\mbox{UpdateEstimations}(\pi, a)$ 
performs the estimation steps of the learning
algorithms, and $\mbox{UpdateProbDist}(u, \sigma'', \tau)$ is the
smooth best response decision rule defined in
Equation~(\ref{eq:smbr}) or best response in Equation~(\ref{eq:BR}). The probability of a transition executing
action $a$ comes from $\sigma(a)$. 

In principle, there could be
a large, even infinite, number of states for certain algorithms due to
the way of updating the parameters of these algorithms. 
For example, the simulation in Figure~\ref{eq:dtmctree}
can be extended to infinite iterations.
However, the same behaviour
can be observed from the two grey nodes in this figure according to
Equation~(\ref{eq:kappa}) and (\ref{eq:fp1p}), although these two nodes are not
the same because of the weight functions. Thus, these nodes share some similarity between them,
and there is no need to explore both nodes because of the
same behaviour. Exploration of one node is sufficient.
This observation leads to the key novelty in
our stochastic model checker: {\em behaviour similarity relation}
among states, which always results in a succinct model and
makes the verification possible.

\subsection{Behaviour similarity relation} \label{subsec:similarity}

% When two states $s_1$ and $s_2$ are behaviour similar, denoted by $s_1\preceq
% s_2$, every successor state of $s_1$ is behaviour similar to a
% successor state of $s_2$, and vice versa. Therefore, only one state
% between $s_1$ and $s_2$ needs to be explored further, as shown in Algorithm~\ref{algo2}. The
% definition of behaviour similarity, shown in Definition~\ref{def:similarity}, varies between learning
% algorithms. 
%

%% We should mention here that even though in a general
%% $0\leq\epsilon_{1}\leq1$, $\epsilon_{1}\sim 1$ cannot be the case
%% under the scope of the behaviour similarity rule for game-theoretic
%% learning algorithms. This is because if a very large $\epsilon_{1}$ is
%% used as $\epsilon_{1}\sim 1$ all the states will be automatically
%% similar. Thus in the rest of this paper we will assume that
%% $\epsilon_{1}\in [0,0.01]$, and in particular, we will use
%% $\epsilon_{1}=10^{-4}$.

\begin{definition}\label{def:similarity_FP}%[Similarity relation]
For two state $s_1=\langle \sigma_1, \pi_1\rangle$ and $s_2=\langle
\sigma_2, \pi_2\rangle$ such as $s_1$ is generated earlier than
$s_2$ by Algorithm~\ref{algo2} such that $s_1$ is generated at the
$t_1$-th iteration and $s_2$ at the $t_2$-th iteration ($t_1\le t_2$), $s_1$ is behaviour similar to $s_2$
under FP, 
%%%%%%%%%%%%JSFP, 
GFP and AFFFP with best response, denoted as
$s_1\preceq s_2$, if the following condition holds:
\begin{equation}
%s_1\preceq s_2\equiv \sigma_1\preceq \sigma_2\land \pi_1\preceq \pi_2,
s_1\preceq s_2\equiv (\sigma_1 = \sigma_2)
\land \Big(\bigwedge_{1\le i\in \mathcal{I}}r_{1, t_1}^i\precsim r_{2, t_2}^i\Big),
\end{equation}
where $r_{k, t_k}^i$ ($k=1,\ 2$) is the expected rewards that player $i$ has
over its actions $A^i$.

Let $\sigma_k^i$ ($k=1,2$) be the strategy of player $i$ in state $s_k$, 
$a^i_{t_k}\in A^i$ the action chosen to be executed in $s_k$, and
$a_{t_k}\in A$ the joint action executed in $s_k$. Note
that $a_{t_1} = a_{t_2}$, and therefore, $a^i_{t_1} = a^i_{t_2}$.
Given a path $\rho=s_0'\cdots s_m'$ such that $s_0'$ is obtained at the
$t_0'$-th iteration, let $\omega(\rho)=a_{t'_0}\cdots a_{t'_0+m-1}$ be the
sequence of joint actions such that $a_{t_0'}$ is executed in $s_j'$
($0\le j\le m-1$).
Let $\check{s}_k$ be the predecessor of $s_k$, $\check{\sigma}_k$ the probability
distribution in $\check{s}_k$, and $\check{a}_k$ the joint action
selected in $\check{s}_k$. 
We define $r_{1,t_1}^i\precsim r_{2, t_2}^i$ iff all the following conditions holds.
\begin{enumerate}
\item[(\ref{def:similarity_FP}.1)] %If  $r_1^i\precsim r_2^i$ then:
\begin{eqnarray}
\mathop{\rm argmax}_{a^{i} \in A^i} \quad r^{i}_{1,t_{1}}(a^i)-r^{i}_{1,t_{1}-1}(a^i)= \nonumber \\
\quad \mathop{\rm argmax}_{a^{i} \in A^i} \quad r^{i}_{2,t_{2}}(a^i)-r^{i}_{2,t_{2}-1}(a^i) 
\label{cond1}
\end{eqnarray}

\item[(\ref{def:similarity_FP}.2)] 
If $\check{\sigma}_1=\check{\sigma}_2=\sigma_1$, then 
$\forall a^i\in
A^i$ such that $\sigma_k^i(a^i)=0$, we either have 
\begin{equation}
r_{2, t_2}^i(a^i)-r_{1,
  t_1}^i(a^i)\le r_{2, t_2}^i(a^i_{t_2})-r_{1, t_1}^i(a^i_{t_1}) \label{cond2}
\end{equation} 
or
\begin{equation}
r^i(a^i_{t_1}, a^{-i}_{t_1})=\max_{a^i\in A^i}{r^i(a^i, a^{-i})}. \label{cond3}
\end{equation}
 
\item [(\ref{def:similarity_FP}.3)]
If there is a path $\rho=s_0'\cdots s_n'$ such that $s_0'=s_1$ and
$s_n'=s_2$ and $n\ge 1$, then $\forall j\in \{1,\cdots, n\}$,
\begin{eqnarray}
\sigma_{1, j}=\sigma_{2, j}\land \bigwedge\limits_{1\le i\le \mathcal{I}}r_{2, t_2}^i(a^i_{t_2+j})\le
  r_{1, t_1}^i(a^i_{t_1+j}), \label{cond4} \\ \textrm{ when FP is used} \nonumber
\end{eqnarray}
\begin{eqnarray}
  \sigma_{1, j}=\sigma_{2, j}\land \bigwedge\limits_{1\le i\le \mathcal{I}}r_{2, t_2}^i(a^i_{t_2+j})\ge
  r_{1, t_1}^i(a^i_{t_1+j}), \label{cond6} \\ \textrm{ when GFP and AFFFP 
is used} \nonumber
\end{eqnarray}
where $\sigma_{k, j}$ ($k=1,2$) is the joint strategy obtained by executing a
sequence of joint actions $\omega(s_0'\cdots s_j')$ from $s_k$.

\item [(\ref{def:similarity_FP}.4)]
If $n=1$, then
\begin{equation}
\bigwedge_{1\le i\le \mathcal{I}}r_{2, t_2}^i(a^i_{t_2})\ge r_{1, t_1}^i(a^i_{t_1}). \label{cond5}
\end{equation} 

\item [(\ref{def:similarity_FP}.5)]
If there is no path from $s_{1}$ to $s_{2}$, then
$\check{\sigma}_1=\check{\sigma}_2$ and
\begin{eqnarray}
\bigwedge_{1\le i\in \mathcal{I}}r_{2, t_2}^i(a^i_{t_2})\ge r_{1,
  t_1}^i(a^i_{t_1}), \label{cond7} \\ 
\nonumber \textrm{ when FP is used}
\end{eqnarray}
\begin{eqnarray}
\bigwedge_{1\le i\in \mathcal{I}} \Big(r_{2, t_2}^i(a^i_{t_2})\ge r_{1,
  t_1}^i(a^i_{t_1}) \land \\
\nonumber \bigwedge\limits_{a^i\in A^i, a^i\not=a^i_{t_1}} r_{1,t_1}^i(a^i)<
r_{2, t_2}^i(a^i)\Big),\label{cond8} 
\end{eqnarray}
 $\quad \quad \quad$ when GFP and AFFFP 
%%%%%%%%%%%%%%%%and JSFP 
is used.

%%  $\forall a^i\in
%% A^i$ such that $\sigma_k^i(a^i)=0$:
%% \begin{equation}
%% r_{1,t_1}^i(a^i)\le r_{2, t_2}^i(a^i) \label{cond6}
%% \end{equation} 

\end{enumerate}
\end{definition}

\subsection{Linearity of expected reward in FP type algorithms.} \label{subsec:linearity}
In this subsection, we show the linear relation of the expected rewards of a player when Fictitious Play based
learning algorithms are used. The results of this section will be used in developing the behaviour similarity
relation for these algorithms. 

Given two states $s_1$ and $s_2$ with
$\sigma_1=\sigma_2$, where the same action
$a$ is played in both states of $s_{1}$ and
$s_{2}$. Assume $s_1$ and $s_2$ are generated at the $t_1$-th and
$t_2$-th iteration respectively with $t_1\le t_2$. Updates of player $i$'s estimations of its opponents'
strategies in FP, GFP, and AFFFP 
%%%%%%%%%%%%and JSFP 
will be the following, $\forall j \in \mathcal{I}$:
\begin{itemize}
\item Fictitious play updates:\\ 
%\begin{minipage}{0.4\textwidth}
\begin{displaymath}
\sigma^{i \rightarrow j}_{t_{1}+1}=\left\lbrace\begin{array}{ll}
(1-\frac{1}{t_{1}+2})\sigma^{i\rightarrow j}_{t_{1}}(a^{j})+\frac{1}{t_{1}+2} & \textrm{ if } a^{j} \in a \\
(1-\frac{1}{t_{1}+2})\sigma^{i\rightarrow j}_{t_{1}}(a^{j}) & \textrm{ if } a^{j} \not \in a
\end{array} \right. ,
\end{displaymath}
%\end{minipage}
%\begin{minipage}{0.6\textwidth}
\begin{equation}
\sigma^{i \rightarrow j}_{t_{2}+1}=\left\lbrace\begin{array}{ll}
(1-\frac{1}{t_{2}+2})\sigma^{i\rightarrow j}_{t_{2}}(a^{j})+\frac{1}{t_{2}+2} & \textrm{ if } a^{j} \in a \\
(1-\frac{1}{t_{2}+2})\sigma^{i\rightarrow j}_{t_{2}}(a^{j}) & \textrm{ if } a^{j} \not \in a
\end{array} \right.
\label{prfp1}
\end{equation} 
%\end{minipage}
\item Geometric fictitious play updates:\\ 
%\begin{minipage}{0.4\textwidth}
\begin{displaymath}
\sigma^{i \rightarrow j}_{t_{1}+1}=\left\lbrace\begin{array}{ll}
(1-\alpha)\sigma^{i\rightarrow j}_{t_{1}}(a^{j})+\alpha & \textrm{ if } a^{j} \in a \\
(1-\alpha)\sigma^{i\rightarrow j}_{t_{1}}(a^{j}) & \textrm{ if } a^{j} \not \in a
\end{array} \right. ,
\end{displaymath}
%\end{minipage}
%\begin{minipage}{0.6\textwidth}
\begin{equation}
\sigma^{i \rightarrow j}_{t_{2}+1}=\left\lbrace\begin{array}{ll}
(1-\alpha)\sigma^{i\rightarrow j}_{t_{2}}(a^{j})+\alpha & \textrm{ if } a^{j} \in a \\
(1-\alpha)\sigma^{i\rightarrow j}_{t_{2}}(a^{j}) & \textrm{ if } a^{j} \not \in a
\end{array} \right.
\label{prgfp1}
\end{equation}
%\end{minipage}
\item Adaptive forgetting factor fictitious play updates:\\ 
%\begin{minipage}{0.4\textwidth}
\begin{displaymath}
\sigma^{i \rightarrow j}_{t_{1}+1}=\left\lbrace\begin{array}{ll}
(1-\alpha_{t_{1}})\sigma^{i\rightarrow j}_{t_{1}}(a^{j})+\alpha_{t_{1}} & \textrm{ if } a^{j} \in a \\
(1-\alpha_{t_{1}})\sigma^{i\rightarrow j}_{t_{1}}(a^{j}) & \textrm{ if } a^{j} \not \in a
\end{array} \right. ,
\end{displaymath}
%\end{minipage}
%\begin{minipage}{0.6\textwidth}
\begin{equation}
\sigma^{i \rightarrow j}_{t_{2}+1}=\left\lbrace\begin{array}{ll}
(1-\alpha_{t_{2}})\sigma^{i\rightarrow j}_{t_{2}}(a^{j})+\alpha_{t_{2}} & \textrm{ if } a^{j} \in a \\
(1-\alpha_{t_{2}})\sigma^{i\rightarrow j}_{t_{2}}(a^{j}) & \textrm{ if } a^{j} \not \in a
\end{array} \right.
\label{prafffp1}
\end{equation}
%\end{minipage}
\end{itemize}

The expected reward of each action $a^{i}$ for state $s_{i}$ $(i
=1,2)$ is computed as follows when joint action $a$ has been played: 
\begin{equation}
r^{i}_{t_{i}}(a^{i}, \sigma_{t_{i}}^{-i}(a^{-i}))=\sum_{a^{j} \in
 A^{-i}} \big(r(a^{i},a^{j})\times\prod_{k:a^{k}\in
 a^{j}}\sigma_{t_{i}}^{i \rightarrow k}(a^{k})\big). \label{eq:pr}
\end{equation}
The elements of $\prod_{k:a^{k}\in a^{j}}\sigma_{t_{i}}^{i \rightarrow
 k}(a^{k})$ for FP, GFP and AFFFP are updated the above stated
estimations. 
The following Lemmas can be induced:
\begin{lemma}
\label{Lem:lin1}
Assume state $s$ is generated in the
$t$-{th} iteration of a learning process and the joint action $a$ is
played in $s$. Let
$a^{-i}(j)$ be the action of player $j$ ($j\not=i$) in $a^{-i}$,
 and $\sigma^{i \rightarrow j}_{t}(a^{-i}(j))$ the
element of $\sigma^{i\rightarrow j}_{t}$ which corresponds to the
action $a^{-i}(j)$. For all actions in $A^i\setminus \{a^{-i}(j)\}$, the rank among them in
$\sigma^{i\rightarrow j}_{t}$ is maintained in
$\sigma^{i\rightarrow j}_{t+1}$, i.e., if
$\sigma^{i\rightarrow
  j}_{t}(a^{j}_1)>\sigma^{i\rightarrow
  j}_{t}(a^{j}_2)$, then $\sigma^{i\rightarrow
  j}_{t+1}(a^{j}_1)>\sigma^{i\rightarrow
  j}_{t+1}(a^{j}_2)$, $\forall a^j_1, a^j_2 \in
A^{j} \setminus \{a^{-i}(j)\}$.
\end{lemma}
\begin{proof}
Player $i$ updates its estimates about the
strategies of its opponents using $a^{-i}$. Then $\sigma_{t+1}^{i\rightarrow j}(a^{-i}(j))$ will
be increased by $\frac{1}{t}(1-\sigma_{t}^{i\rightarrow
  j}(a^{-i}(j)))$ and for all $a^j\in A^j\setminus \{a^{-i}(j)\}$, the estimate
$\sigma_{t+1}^{i\rightarrow
  j}(a^{j})$ will be
decreased by $\frac{1}{t}\sigma_{t}^{i\rightarrow
  j}(a^{j})$. For any two numbers $x$ and $y$, we have
$x>y\land t>0
\Leftrightarrow (1-\frac{1}{t})x> (1-\frac{1}{t})y$.
\end{proof}

\begin{lemma}
\label{Lem:lin2}
Assume that the same joint action $a$ is played in any two consecutive
iterations, i.e., $t$-th and $(t+1)$-th iterations, using
FP, GFP or AFFFP. Given two actions $x^{i}\in A^i$ and $y^{i}\in A^i$
($y^i\not= x^i$) of player
$i$ of which $x^{i}$ is the best response to $a$, let 
\begin{eqnarray}{rcl}
\nonumber
w^{\dagger} & = & r^{i}_{t+1}(x^{i}, \sigma^{-i}_{t+1}(x^{-i}))-r^{i}_{t}(x^{i},\sigma^{-i}_{t}(x^{-i})),\\
\nonumber
w^{\ddagger} & = & r^{i}_{t+1}(y^{i},\sigma^{-i}_{t+1}(y^{-i}))-r^{i}_{t}(y^{i},\sigma^{-i}_{t}(x^{-i})),\\
\nonumber
\bar{w}^{\dagger} & = & r^{i}_{t+2}(x^{i},\sigma^{-i}_{t+2}(x^{-i}))-r^{i}_{t+1}(x^{i},\sigma^{-i}_{t+1}(x^{-i})),\\
\nonumber
\bar{w}^{\ddagger} & = & r^{i}_{t+2}(y^{i},\sigma^{-i}_{t+2}(y^{-i}))-r^{i}_{t+1}(y^{i},\sigma^{-i}_{t+1}(y^{-i})).
\end{eqnarray}
We have $w^{\dagger}>w^{\ddagger}$ and $\bar{w}^{\dagger}>\bar{w}^{\ddagger}$.
\end{lemma}
\begin{proof}
The fact that $x^{i}$ is the best response to $a$ indicates that if
$\sigma^{-i}(a^{-i})=1$, then player $i$ will choose action $x^{i}$ at
the next iteration. If $\sigma^{-i}(a^{-i})<1$, then there exists a
threshold $0<f<1$ such that player $i$ will continue to choose action
$a^{i}$ when $\sigma^{-i}(a^{-i})<f$; player $i$ will choose $x^i$
when $\sigma^{-i}(a^{-i})\ge f$
\cite{learning_in_games}. 
As the expected rewards are linear functions, we have
$\bar{w}^{\dagger}>0$, which means that player $i$'s confidence on
selecting $x^i$ is increasing as $a$ is played again. This confidence
will keep increasing as $\sigma^{-i}(a^{-i})$ increases, until $\sigma^{-i}(a^{-i})\ge f$ and then $x^i$ is
played.
If there exists another action $y^{i}\not= x^i$ whose expected reward
is also increased at the $(t+1)$-th iteration, then we have
$w^{\ddagger} > 0$. If its expected reward is increasing faster than $x^{i}$, i.e.,
$w^{\dagger}<w^{\ddagger}$ and $\bar{w}^{\dagger}<\bar{w}^{\ddagger}$, then at a later iteration of any
of the game playing process FP, GFP and AFFFP, it is possible to select
$y^i$ if $a$ is played sufficient number of times to reach
$\sigma^{-i}(a^{-i})\ge f$. However, this is a contradiction because the best
response to $a$ is $x^{i}$ and not $y^{i}$. 
Furthermore, we cannot have $w^{\dagger}<w^{\ddagger} \land
\bar{w}^{\dagger}\ge\bar{w}^{\ddagger}$ or $w^{\dagger}\ge w^{\ddagger} \land
\bar{w}^{\dagger}<\bar{w}^{\ddagger}$ because of the expected rewards are linear functions with respect to the estimates of other players' strategies.
\end{proof}

\subsection{State space reduction} \label{subsec:reduction}

Now we can explain what {\em similarity} means for FP based learning
algorithms. First, we need to define {\em BSCC actions} as
follows. Let $\sigma_s$ denote that the strategy $\sigma$ in a state
$s =\langle \sigma, \pi \rangle$.
\begin{definition}
Given a BSCC $bscc$, the corresponding BSCC actions is the set of
actions $A_{bscc} \subseteq A$ such that 
$$A_{bscc} = \{a\in A\mid \exists s\in bscc \mbox{ such that }
 \sigma_s(a)>0 \}.$$
\end{definition}

\begin{theorem}
\label{prop:FP}
Given two states $s_1$ and $s_2$ such that $s_1\preceq s_2$ under FP,
GFP, and AFFFP,
%%%%%%%%%%%% and JSFP, 
if a set of BSCC actions $A_{bscc}$ can be reached from
$s_{1}$, then $A_{bscc}$ will be reached from $s_2$.
\end{theorem}

\begin{proof}

In order to complete this proof it is needed to show that when two states, $s_{1}$ and $s_{2}$ are found to be similar then the same actions will be selected in their successor states and eventually they will converge to the same BSCC. We should notice here that it is not necessary for both states to need the same number of successor states in order to reach a BSCC. 

Suppose that $s_1$ is generated at iteration $t_1$ and $s_2$ at $t_2$
and we have $t_1\le t_2$. For all the variants of FP, that $s_{1}$ and $s_{2}$ are found similar belong to one of the following four cases:
\begin{enumerate}
\item  There is no path from $s_{1}$ to $s_{2}$ and
  $s_{1}$ itself is a BSCC;
\item  There is no path from $s_{1}$ to $s_{2}$ and
  $s_{1}$ alone is not a BSCC;
\item  $s_{2}$ is the direct successor state of $s_1$; %(This is also a criterion to find if $s_{1}$ is a BSCC)
\item  There are intermediate states between $s_{1}$ and $s_{2}$. %(This is also a criterion to find if $s_{1}$ is a BSCC)
\end{enumerate}

\underline{Case 1:} \\
As $s_1$ is a single state BSCC, the joint action $a$
selected in $s_1$ will be played infinitely from $s_1$ onward, which
implies that a Nash equilibrium is reached by playing $a$. By definition of {\em
  behaviour similarity}, the same joint
action $a$ will be also played in $s_{2}$, and therefore, the same
Nash equilibrium is reached. By the definition of Nash equilibrium, we
know that from $s_2$, action $a$ will be played permanently.

\underline{Case 2:} \\ Similarly to Case 1 since $s_1\preceq s_2$,
conditions (6.1), (6.2) and (6.5) should be satisfied. This case can
be divided in two sub-cases. The first one consists of the instances
where $\check{\sigma}_1=\check{\sigma}_2 \not = \sigma_1$, which means that
the joint actions $\check{a}_k$ ($k=1,2$) from the predecessors of $s_{1}$ and $s_{2}$ are the
same, but they are different from the
action $a_{t_k}$ selected in $s_{1}$ and
$s_{2}$. In both $s_1$ and $s_2$, the selected joint action changes from $\check{a}_k$ to its
best response $a_{t_k}$, which denotes that the players in both states reach
the necessary confidence level in order to change actions \cite{brown_fict,learning_in_games}. Since $a_{t_k}$
is not a Nash equilibrium, there exists another
joint action $\bar{a}_{t_k}$ that is the best response to $a_{t_k}$.
Furthermore, we have $\bar{a}_{t_1}=\bar{a}_{t_2}$ due to $a_{t_1}=a_{t_2}$ and 
condition (6.5). This condition is used to guarantee that the selected action is not
affected by smooth best response at the initial
iteration. As the best response $\bar{a}_{t_k}^i$ to $a_{t_k}^i$ for player $i$ is
deterministic, if the joint best response $\bar{a}_{t_1}$ is not equal to
$\bar{a}_{t_2}$, then we know that there exists a player $j$ that excises 
its best response earlier than some other players. This asynchronous
change is caused by the smooth best response at the initial iteration:
players chooses a non-best response to their initial
estimation. Condition (6.5) is designed to prevent two states from
being classified similar when asynchronous changes can occur in their
successor states.
Note here
that in the case of GFP and AFFFP, the learning rate will be the same
in both $s_1$ ans $s_2$, as it does not depend on the number of iterations. For that
reason, the extra constraint in Equation (\ref{cond8}) is needed.
When the joint action $\bar{a}_{t_k}$ is selected, the estimates of
opponents strategies will be increased for the elements that
participate in  
the joint action $\bar{a}_{t_k}$. Players will learn that $\bar{a}_{t_k}$ is
selected and eventually they will select the best response to
$\bar{a}_{t_k}$, say $\bar{\bar{a}}_{t_k}$, simultaneously. This
process will continue until a Nash equilibrium or a BSCC is reached.

In the second case, the same action has been played in $s_{1}$,
$s_{2}$ and their predecessors. Similarly to the first case because of
condition (6.1), there is another action $\bar{a}^{i}_{t_k}$ that is best
response to $a^{i}_{t_k}$. Nonetheless, it not necessary that the players
will simultaneously change their action and result to the same joint
action from both states. Assume that exists another action $b^{i}$ in state $s_{2}$ such as
$r^i_{2,t_{2}}(a^{i}_{t_k})-r^i_{2,t_{2}}(b^{i})<r^i_{t_{2}}(a^{i}_{t_k})-r^i_{2,t_{2}}(\bar{a}^{i}_{t_k})$.
This implies that it is possible for $r^{i}_{2,t_{2}}(b^{i})$ to exceed $r^{i}_{2,t_{2}}(a^{i}_{t_{k}})$ in fewer iterations than $r^{i}_{2,t_{2}}(\bar{a}^{i}_{t_{k}})$ will do.
%This implies that given the expected reward of $a_{t_{k}}$ has the greatest
%increment in the expected reward among all actions, if another action
%$b$ can be selected from $s_{2}$, then its expected reward should also
%increasing and its difference to the expected reward of the currently
%selected action is smaller than the one of $a_{t_{k}}$. 
% If the expected reward $r^i_{t_2}(b)$ is also increasing, then it is possible that $b$ will be also selected.
Condition (6.2)
is designed to overcome this problem. This condition ensures that the difference in expected reward
between the selected and the non-selected actions in $s_2$, which
has the slowest learning rate due to $t_2\ge t_1$, is smaller than the one in
$s_{1}$.
% This implies that even in the case where there is also
%another joint action, $b$, that increases its expected reward with
%rate smaller than $\tilde{a}$ it will not be
%selected.
Condition (6.2) implies that
$r^i_{1,t_{1}}(a^{i}_{t_k})-r^i_{1,t_{1}}(b^{i})<r^i_{2,t_{2}}(a^{i}_{t_k})-r^i_{2,t_{2}}(b^{i})$. Therefore,
if $b^{i}$ was to be selected from a successor of $s_2$, then the
expected reward of $b^{i}$ should be increased more from $s_{2}$ than
from $s_{1}$. Therefore, because of Lemma \ref{Lem:lin1} and the
Condition (6.5), it is not possible for $b^{i}$ to be selected from a
successor of $s_{2}$ given that it was not selected from a successor
of $s_{1}$. 

Now we assume that $\bar{a}_{t_k}$ has been played in a successor
  state of $s_1$ and $s_2$ at the $(t_1+n_1)$-th and $(t_2+n_2)$-th
  iterations respectively. We have the following inequalities:
$$r^{i}_{1,t_{1}+n_{1}}(\bar{a}^{i}_{t_k})> r^{i}_{1,t_{1}+n_{1}}(a^{i}_{t_k}) \qquad  r^{i}_{2,t_{2}+n_{2}}(\bar{a}^{i}_{t_k})> r^{i}_{2,t_{2}+n_{2}}(a^{i}_{t_k}).$$
 If $\bar{a}_{t_k}$ is not a Nash
  equilibrium, we assume that $\ddot{a}^i\in A^i$ is the action of robot $i$
  that has the highest increment in the expected rewards such that
$$r^{i}_{1,t_{1}+n_{1}}(\ddot{a}^{i})<
  r^{i}_{1,t_{1}+n_{1}}(\bar{a}^{i}_{t_k}) \qquad
  r^{i}_{2,t_{2}+n_{2}}(\ddot{a}^{i})<
  r^{i}_{2,t_{2}+n_{2}}(\bar{a}^{i}_{t_k}).$$
  
If the expected reward of action $c^{i}\in A^i$ is also increasing, then we
  have:
$$r^{i}_{1,t_{1}+n_{1}}(c^{i})< r^{i}_{1,t_{1}+n_{1}}(\bar{a}^{i}_{t_k})
  \qquad r^{i}_{2,t_{2}+n_{2}}(c^{i})<
  r^{i}_{2,t_{2}+n_{2}}(\bar{a}^{i}_{t_k}).$$
  
Assume that $\ddot{a}^{i}$ will be selected at a successor state of
$s_{1}$. Therefore the expected reward of action $\ddot{a}^{i}$ will
exceed the expected reward of $\bar{a}^{i}_{t_k}$ either in fewer
iterations than that of $c^{i}$ does, or if they need the same number of
iterations, then its expected reward will be greater than that of
$c^{i}$.
 %%	 Then it can be shown that $\ddot{a}^{i}$ will be
 %%  also selected in a successor state of $s_{2}$. As
 %%  $\bar{a}_{t_k}$ has been selected at the $(t_1+n_1)$-th and $(t_2+n_2)$-th
 %%  iterations respectively and condition
 %% (7.2) holds, the following inequalities hold: 
Based on the
linearity of the expected reward with respect to the expected
strategy, Lemma \ref{Lem:lin1}, Lemma \ref{Lem:lin2} and the above inequalities, the same
will also happen in the successor states of $s_{2}$ with $\ddot{a}^{i}$ being
eventually selected. If $\ddot{a}^{i}$ is not a Nash equilibrium, then
the next action will be selected from both states for the same
reason until a Nash equilibrium or a BSCC is reached.

\underline{Case 3:} \\ In this case $s_{2}$ is the immediate successor
of $s_{i}$, i.e. $t_{2}=t_{1}+1$. Similarly to the previous cases
since $s_1\preceq s_2$ conditions (6.1), (6.2) and (6.4) will be satisfied. 
%Equations (\ref{cond1}), (\ref{cond2}) or(\ref{cond3}), (\ref{cond5}) should be satisfied.
If $s_{1}$ is a BSCC then the selected action in all its successor
states will be the same, e.g., $a$. Equation (\ref{cond1}) and
(\ref{cond5}) show that $a$ will be played also in the successor state
of $s_{2}$, since players use best response to select
actions. Although BR does not guarantee that there will not exist
another action $\bar{a}$ that will be selected in a later
state. However, if Equation (\ref{cond3}) is met, then player $i$ will
not like to change its action, as this will be the one which will
maximise its reward. If Equation (\ref{cond2}) is satisfied, then the
difference between the expected reward of selected action and the
actions $a^{i}: a^{i} \in A^{i} \wedge a^{i} \not \in a$ is increasing
in $s_{2}$, as Equation (\ref{cond2}) can be written as $\forall
a^i\in A^i$ such that $\sigma_k^i(a^i)=0$:
\begin{displaymath}
r_{1, t_1}^i(a^i_{t_1})-r_{1,t_1}^i(a^i)\le r_{2, t_2}^i(a^i_{t_2})-r_{2, t_2}^i(a^i) .
\end{displaymath}
Therefore, the expected reward of joint action $a$ will be increased
more than all the other joint actions, and thus, it will be selected in
all the successor states of $s_{1}$ and $s_{2}$. Thus $a$ is a
Nash equilibrium and $s_{1}$ and $s_{2}$ can be characterised as similar states.

\underline{Case 4:} \\ In this case, conditions (6.1), (6.2) and (6.3)
hold, but FP, GFP, and AFFFP are treated differently in condition
(6.3). 
%We expect that the loop $s_0'\cdots s_n' s_0'$ is a BSCC. 
In this loop $s_0'\cdots s_n' s_0'$, the estimations of a player's
opponents strategies, i.e., Equation (\ref{prfp1}), of FP would
converge to the number of occurrences each joint action appears in the
loop. That is, if a joint action $a$ appears in the loop $m$ times
over the $\mathcal{N}=n+1$ actions of the loop, then for each player
$i \in \{1, \cdots, \mathcal{I}\}$, its estimates of the opponents
strategies $\sigma^{i \rightarrow j}(a^{j})$ is
$\frac{m}{\mathcal{N}}$ for all $j \in \{1, \cdots, \mathcal{I}\}
\setminus \{i\}$ and $a^j \in a$ \cite{learning_in_games}. Thus, the
fluctuations of the expected reward will vanish through time. Hence if
the expected reward for a specific joint action of the loop increases,
this denotes that $s_{1}$ is not a BSCC. On the other hand, GFP and
AFFFP do not converge to the mixed Nash equilibrium of a
game. Instead, they will converge to a mixed strategy, which depends
on $\alpha$ and $\lambda$ respectively. The expected reward for the
joint action $a$ will be greater than $\frac{m}{N}$
\cite{learning_in_games,mythesis}.  Therefore, it is expected that the
estimates of the opponents playing the observed actions will increase,
as well as the specific expected rewards, although the increment will
get smaller in every iteration of the loop. Thus, when Equation
(\ref{cond6}) is satisfied, the the same loop is generated from
$s_{1}$ and $s_{2}$ and their successors.
%%%%%%%%%%%% Similarly the
%% expected reward of the JSFP will also increase when a joint action in
%% a Nash equilibrium is
%% observed, which indicates that a later observed state such as $s_{2}$ is expected to
%% have greater expected rewards.
\end{proof}

\section{Case studies} \label{sec:exp}

We have implemented our method in a prototype tool. In this section we present case studies to demonstrate our performance
evaluation method on several case studies.

\subsection{Simple coordination game}
The game described by Equation~(\ref{toy_reward}) is a symmetric game.
In this game, players can easily be trapped in cycles composed of
joint actions $(b_1, a_2)$ and $(b_2, a_1)$. This cycle does not have any
reward. Figures~\ref{fig:toy957} illustrates the DTMCs (with weights omitted) for the
FP, 
%%%%%%%%%JSFP, 
GFP and  AFFFP,algorithms for the example in
Equation~(\ref{toy_reward}) under the initial estimation in
Equation~(\ref{toy_weight}). In each DTMC, the ellipse
nodes represent the initial states. The BSCCs are
shown in nodes with rounded corners. 

\begin{figure*}%[!h]
\centering
\subfigure[FP]{
\includegraphics[scale=0.4]{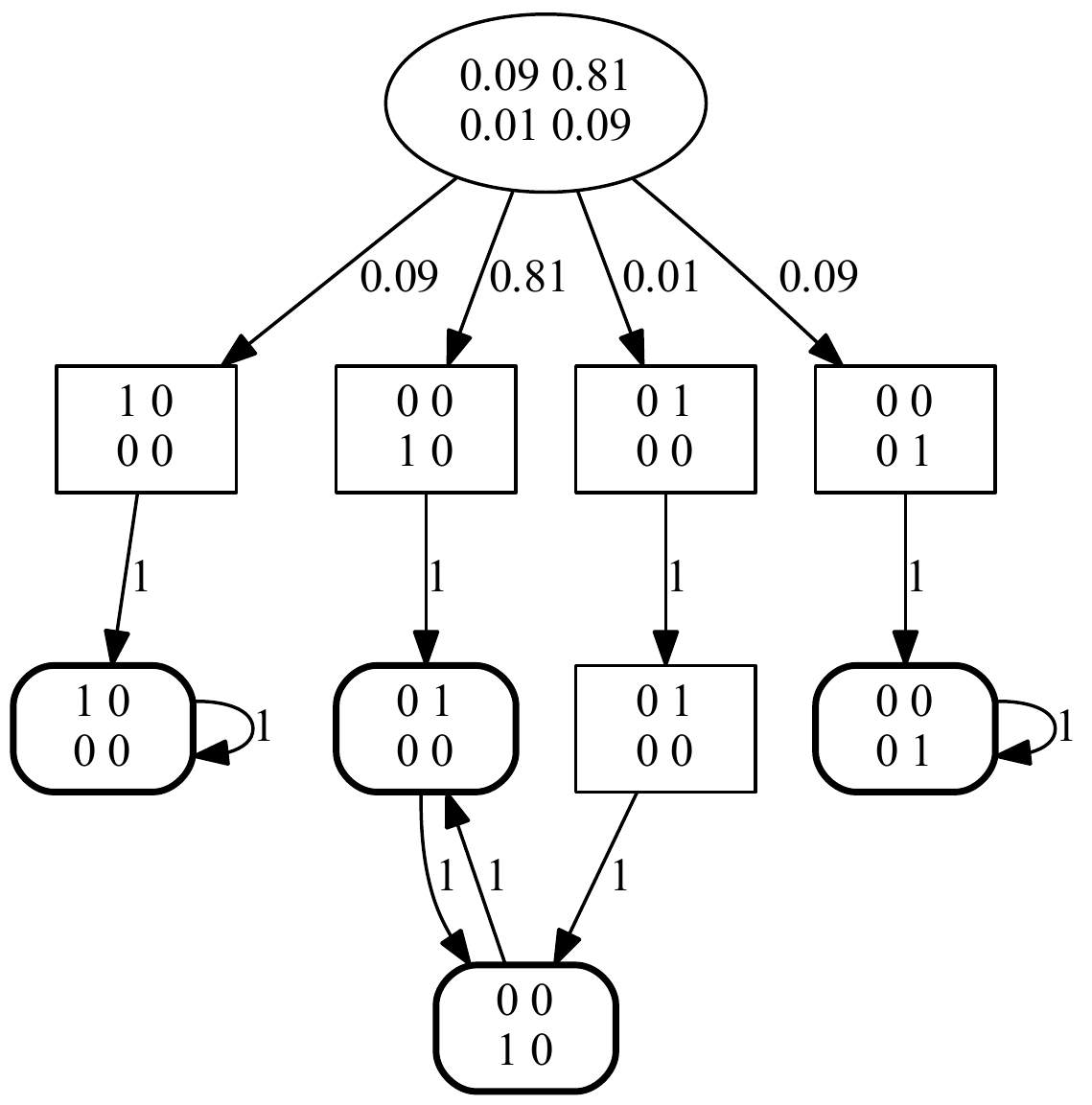}
\label{fig:fp}}
%% \hfil
%% \subfigure[SAP with $\tau=0.0001$]{
%% \includegraphics[scale=0.4]{Toy7Algo2NewCriterion}
%% \label{fig:sap}}
%% \hfil
%% \subfigure[JSFP]{
%% \includegraphics[scale=0.43]{Toy7Algo9NewCriterion}
%% \label{fig:jsfp}}
\hfil
\subfigure[GFP with $\alpha=0.2$ and AFFFP with $\lambda_0=0.8$]{
\includegraphics[scale=0.4]{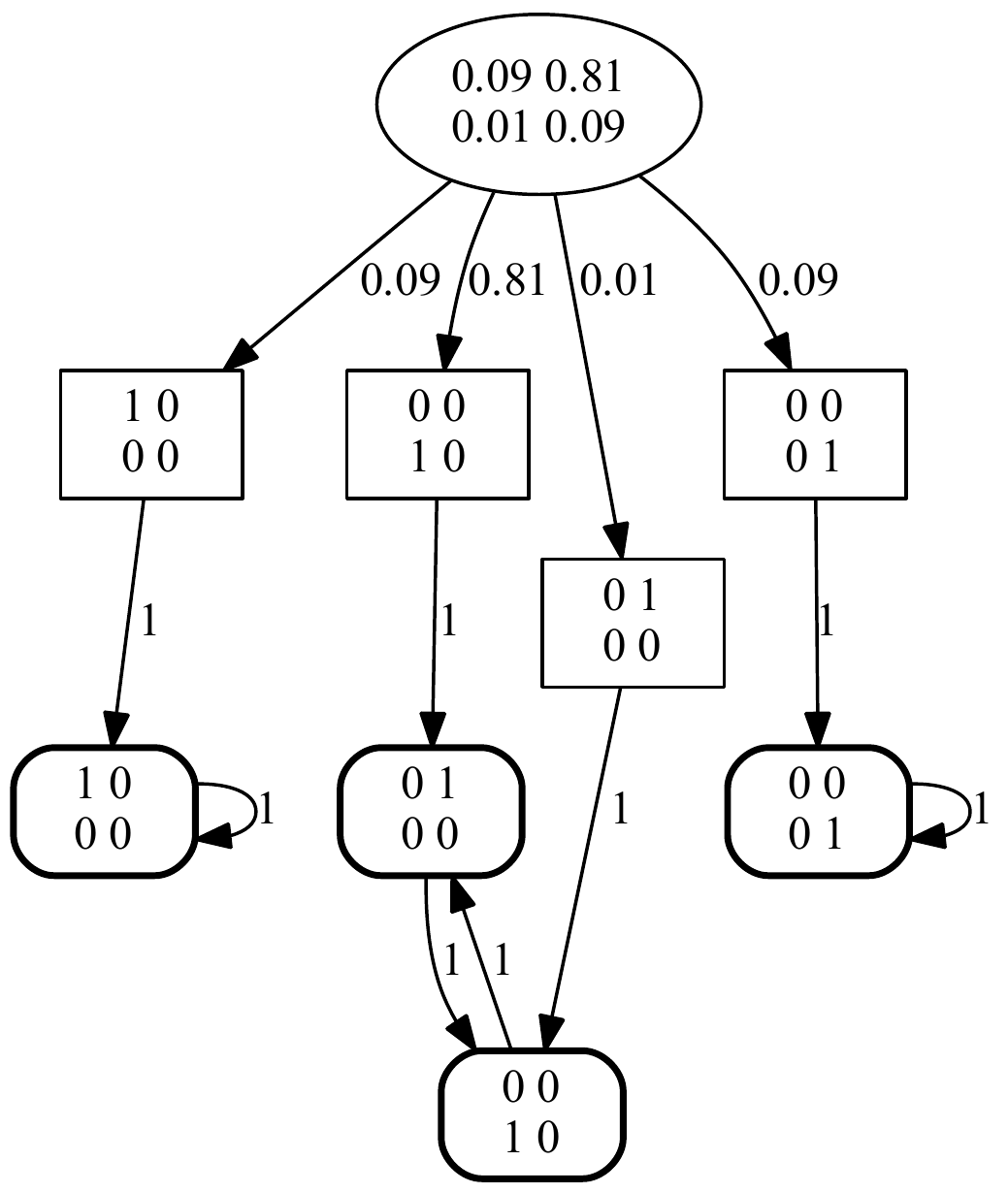}
\label{fig:gfp}}
\caption{The DTMCs for FP, GFP and AFFFP on the simple coordination example}\label{fig:toy957}
\end{figure*}

%\begin{figure*}%[!h]
%\centering
%\includegraphics[scale=0.4]{ToyDSA20062016}
%\caption{The DTMCs for DSA with $p=0.9$ on the simple coordination example}\label{fig:dsa}
%\end{figure*}

%\begin{figure*}%[!h]
%\centering
%\subfigure[RM and MGRM]{
%\includegraphics[scale=0.4]{ToyRM20062016}
%\label{fig:rm}}
%\hfil
%\subfigure[SAP with $\theta=0.1$]{
%\includegraphics[scale=0.4]{ToySAP20062016}
%\label{fig:sap}}
%%% \hfil
%%% \subfigure[JSFP]{
%%% \includegraphics[scale=0.43]{Toy7Algo9NewCriterion}
%%% \label{fig:jsfp}}
%%% \hfil
%%% \subfigure[GFP with $\alpha=0.8$ and AFFFP with $\lambda_0=0.8$]{
%%% \includegraphics[scale=0.4]{ToyGFP08062016}
%%% \label{fig:gfp}}
%\caption{The DTMCs for RM, MGRM and SAP on the simple coordination example}\label{fig:toy958}
%\end{figure*}

These figures clearly illustrate that FP, GFP and AFFFP cannot
avoid this cycle.

Three measures were used to compare the performance of the algorithms:
the number of states, the depth needed to reach these states, and the
probability each algorithm has to converge into Pareto efficient Nash
equilibria.  These are three general measures that can be used in
order to compare the above mentioned algorithms. Depending on the
problem, more measures can be defined in our tool, such as the
probability of a state of interest being reached even if it is not a
Nash equilibrium, or the probability of not reaching a specific state.
However, the exploration of all possible performance measures is not
the focus of this work.  The computational time is not reported
because in the most of the cases were just a fraction of one second.

In the experiments, we set $\tau=1$ for smooth best response used
in the first iteration.
%% , and $\epsilon_1=0.0001$ for checking behaviour
%% similarity in the remaining iterations, except that for RM and MGRM,
%% $\epsilon_1=0.01$, as the small value of $\epsilon_1$ blows up the
%% state space dramatically when running RM and MGRM. 
We generate 100 random initial estimates and run each algorithm over
them to test their average performance, as the initial estimate can
have strong impact on the performance.  The results we obtained for
these games are reported in Table~\ref{table-results1}.  
%This table
%shows that among all algorithms investigated in the paper, RM, MGRM
%and SAP can always converge to the pure Nash equilibria. 
GFP has the highest probability to reach the pure
Nash equilibria among all FP-based algorithms.

\begin{table}%[!h]
\caption{\label{table-results1} Experimental results for the simple
  coordination game with $\tau=1$.}
\centering
\begin{tabular}{|c||r|r|r|r|r|r|r|}
\hline
Performance & FP & GFP & AFFFP \\
\hline
\hline
Number of states & 9 & 64 & 31  \\
\hline
Number of iterations & 2 & 57 & 24 \\
\hline
Convergence probability & 0.7204 & 0.8313 & 0.6666  \\
% \cline{2-12}
%  & Time (s) & $<$0.001 & $<$0.001 & $<$0.001
% & $<$0.001 & 0.003 & 0.029 & $<$0.001 & 0.118 & 0.076 & $<$0.001 \\
\hline
\end{tabular}
\end{table}

\subsection{Complex coordination game}

This game was introduced in \cite{big_game} in order to show the asymptotic properties of FP. 
It is also a symmetric game where the rewards for each player
is the same. The rewards are generated according to
Table~\ref{tab:big_game_rewards}. Each player has 20 actions in this
example, but it can be extended to a larger number of actions. The
full reward matrix can be found in Appendix.
\begin{table}%
\caption{Process to create the reward matrix of the complex
  coordination game where each player has 20 available actions to
  choose with $\delta=0.001$.}
\label{tab:big_game_rewards}
\centering
\begin{tabular}{l}
\hline
$\cdot$ set $n=5$, $\zeta=1+\frac{1}{n^{1-\delta}}$, $\beta=1-\frac{1}{n^{2*(1-\delta)}}$\\
$\cdot$ $u(i,j)=1$, $\forall i \in [n+1,4n], \textrm{ } j=i$\\
$\cdot$ $u(i,j)=1$, $\forall i \in [2,n], \textrm{ } j=i-1$\\
$\cdot$ $u(i,j)=\zeta$, $\forall i \in [n+1,4n], \textrm{ } j=i-1$\\
$\cdot$ $u(i,j)=\zeta$, $i=2n+1$, $j=4n$\\
$\cdot$ $u(i,j)=\beta$, $\forall j\leq 2n, \textrm{ } i>j$\\
$\cdot$ $u(i,j)=\beta$, $\forall i-j\leq n, \textrm{ } i>j$\\
$\cdot$ $u(i,j)=\beta$, $\forall i \in [2n+1,j-n], \textrm{ } j \in [3n+1,4n]$\\
$\cdot$ $u(i,j)=0$, Otherwise\\
\hline
\end{tabular}
\end{table}

It was shown in \cite{big_game} that for any $\delta >0$, FP will always converge to a
solution with a reward no less than $\xi$, where $\xi$ is defined as
the utility of the Pareto efficient Nash equilibrium subtracting a
constant $\epsilon=\frac{1}{2}$. Furthermore, this result was obtained for a single initial joint action
$(a_1,b_1)$.

In this work, we study the performance of FP under random initial
conditions with $0<\kappa_{t}^{i\leftarrow j}(a^{j})\leq 1, \forall
a^{j}\in A^{j}, \forall j \in \{1,\ldots, \mathcal{I}\}$ and $\delta=0.001$. We
found that for this particular $\delta$ and initial conditions, FP
almost surely converges to the Pareto efficient Nash equilibrium. This
supports the importance of the proposed methodology, since even in
cases where the theoretical properties of a game are well known, there
are specific conditions which can be of interest in practical
applications where the outcome of FP will defer.

Additionally, under specific initial conditions there is a temporary cycle
between two joint actions, e.g.,  $a^{c1}=(b_{20},a_6)$ and
$a^{c2}=(b_6,a_{19})$, which disappears after the $t^{th}$
iteration. Figures \ref{fig:util1} and \ref{fig:util2} depict the
expected reward of each action of player 1 and 2 respectively. In both
figures the expected rewards of the actions in support of $a^{c1}$ and
$a^{c2}$ are similar in the initial iterations. However, after $t>500$
iterations, FP converge to a Pareto efficient Nash equilibrium, which
is a single
action whose expected reward is maximised. The reason for this
behaviour is illustrated in Figures \ref{fig:prob1} and \ref{fig:prob2},
where Players estimates of their opponent's strategy are depicted after
$100$ and $2000$ iterations respectively. In Figure \ref{fig:prob1},
there are fluctuations in Player $1$'s estimates that leads to the
cycles between the joint actions. These fluctuations disappear in
Figure \ref{fig:prob2} when FP converged to the Nash equilibrium.

This behaviour was captured by the stopping criterion we used for
FP. When 100 iterations were used the proposed methodology does not
produce a terminal state. Thus, based on the fluctuations of the
estimates about players' opponents' strategies, and consecutively
their expected rewards, there were no evidence of a persistent cycle or
a Pareto efficient Nash equilibrium. When we increase the number of
iterations to $3000$, the
Pareto efficient Nash equilibrium was identified as a terminal state.

\begin{figure}
\centering
\subfigure[Player 1]{
\includegraphics[scale=0.4]{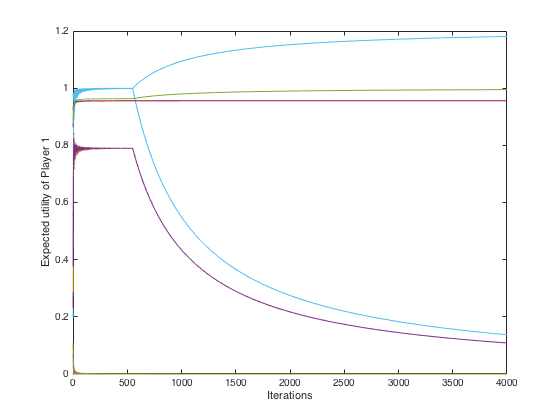}
\label{fig:util1}}
\hfill
\subfigure[Player 2]{
\includegraphics[scale=0.4]{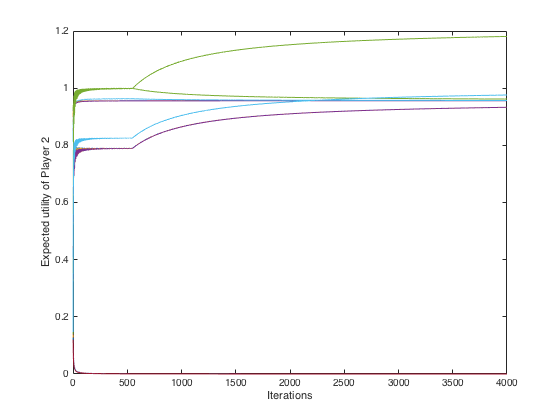}
\label{fig:util2}}
\caption{Expected rewards in the complex coordination game}
\end{figure}

\begin{figure}
\centering
\subfigure[80 iterations]{
\includegraphics[scale=0.4]{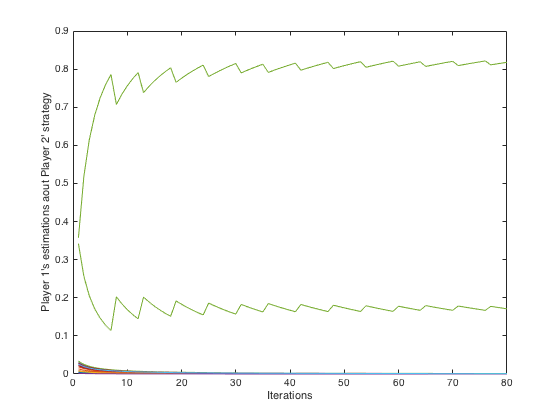}
\label{fig:prob1}}
\hfill
\subfigure[400 iterations]{
\includegraphics[scale=0.4]{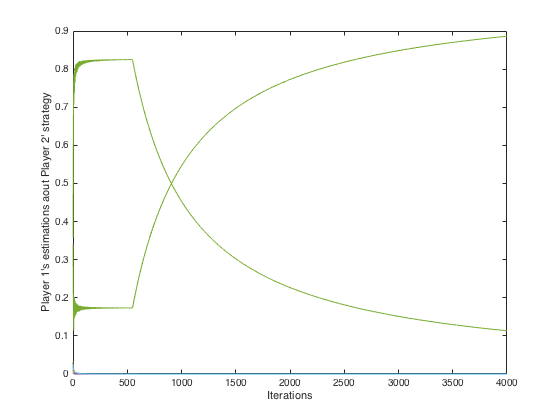}
\label{fig:prob2}}
\caption{Estimates of the opponent' strategy in the complex coordination game}
\end{figure}

%Table~\ref{table-results11} shows the experimental results with $\tau=1$. 
\begin{table}%[!h]
\caption{\label{table-results11} Experimental results for the complex
  coordination game with $\tau=1$.}
\centering
\begin{tabular}{|c||r|r|r|r|r|r|r|}
\hline
Performance & FP & GFP & AFFFP \\
\hline
\hline
Number of states  & 2373 & 2678 & 2271 \\
\hline
Number of iterations  & 90 & 7 & 32\\
\hline
Convergence probability & 0.9969 & 0.999 & 0.9115 \\
% \cline{2-12}
%  & Time (s) & 0.03 & 0.021 & 0.005
% & 0.017 & 1.623 & 101.6 & 0.113 & 0.009 & 0.003 & 0.819 \\
\hline
\end{tabular}
\end{table}

In this example, FP, GFP and AFFFP have good performance. Their convergence
probability is very close to 1. 
%%%%%%%%%%% JSFP, on the other hand, cannot cope with
%% this example well. 
%On the contrary, RM, MGRM and SAP can barely reach the pure Nash
%equilibria. This is because we used deterministic version of these
%three algorithms.
From Tables~\ref{table-results1} and \ref{table-results11},
we conjure that GFP has great potential to be applied in practice for
its high convergence probability and fast convergence.

\subsection{Shapley's game}

Games can often be classified in two categories based on the structure of
their reward function: {\em cooperative} and {\em non-cooperative}
games. The main difference between these two categories is that in
cooperative games, there is a joint action that maximises all players'
reward, while in non-cooperative games, no joint actions can maximise
all players' reward. In the latter case, 
we selected a joint action that is ``fair'' for both players.  
\begin{equation} \label{Shappley_reward}
r=\kbordermatrix{ &a_1& &a_2& &a_3\\
b_1&0,0&\vrule&1,0&\vrule&0,1\\
b_2&0,1&\vrule&0,0&\vrule&1,0\\
b_3&1,0&\vrule&0,1&\vrule&0,0
}
\end{equation}

Equation~(\ref{Shappley_reward}) shows the Shapley's
game~\cite{shapley}, which does not have any pure Nash equilibria.
However, there are two mixed strategy equlibria. In the first mixed
strategy equilibrium, the players
choose with the same probability the joint actions $(b_1, a_1)$,
$(b_2,a_2)$ and $(b_3,a_3)$. In the second one,
the players choose actions based on the following cycle: 
\begin{eqnarray}
(b_1,a_2)
\rightarrow (b_1,a_3) \rightarrow (b_2,a_3) \rightarrow (b_2,a_1)\\
\nonumber \rightarrow (b_3,a_1) \rightarrow (b_3,a_2) \rightarrow (b_1,a_2). \label{eq:shapley-cycle}
\end{eqnarray}

 Figure \ref{fig:ex1} depicts the probability of players to reach a
 decision in Shapley's game in Equation~(\ref{Shappley_reward}) when
 FP is used. In particular, each colour represent one of the joint
 actions that comprises the cycle in
 Equation~(\ref{eq:shapley-cycle}). From Figure \ref{fig:ex1}, it is
 reasonable to assume that FP converges to a single action because
 this joint action is played for more than a thousand iterations of
 the game. However, it is well known \cite{shapley} that FP is trapped
 in this cycle, where each of the six joint actions in this cycle is
 repeated for a large, but still finite, number of iterations as
 $t\rightarrow \infty$. In particular, the number of repetitions for
 each joint action in a loop is increased as $t\rightarrow
 \infty$. Our tool does not terminate on this cycle because of its
 irregular shape. 
\begin{figure}
\centering
  \includegraphics[scale=0.4]{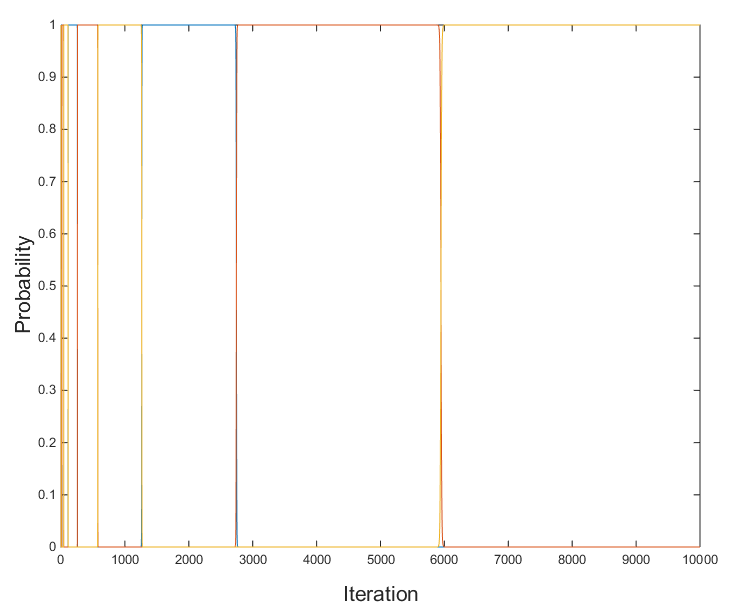}
  \captionof{figure}{Simulation for Shappley's game}   \label{fig:ex1}
\end{figure}

However, our tool successfully captures the cycle in the first mixed
strategy equilibrium. This cycle can only be generated when both
players use equal weights as the initial estimates, which is shown in
Equation~(\ref{shapley_weight}). 
\begin{equation} \label{shapley_weight}
\kappa_{0}^{1\rightarrow 2}=[\frac{1}{3}, \frac{1}{3},
  \frac{1}{3}]^{T} \mbox{ and } \kappa_{0}^{2 \rightarrow 1}=[\frac{1}{3}, \frac{1}{3},
  \frac{1}{3}],
\end{equation}
It is also worth noting that if two or
more actions have the same expected rewards when computing
Equation~(\ref{eq:BR}) in best response, which is the case in this mixed
strategy equilibrium, then the first action will be chosen
to execute. The DTMC showing the cycle is presented in
Figure~\ref{fig:shapleydtmc}, where the states and transitions for the
second equilibium is omitted. Note that the prefix before the cycle in
the middle and right branches in this figure is generated due to the
smooth best response rule in the first iteration. Under best response
and equal initial weights in Equation~(\ref{shapley_weight}), the
players would definitely choose $(b_1,a_1)$, i.e., the left branch in the figure,
while smooth best response can force them to play $(b_2,a_2)$ (the
middle branch) and $(b_3,a_3)$ (the right branch) with non-zero probability.
These two branches
enter the cycle after the weights of all actions become equal again
and this time they use best response.

\begin{figure}
\centering
\subfigure[FP]{
\includegraphics[scale=0.4]{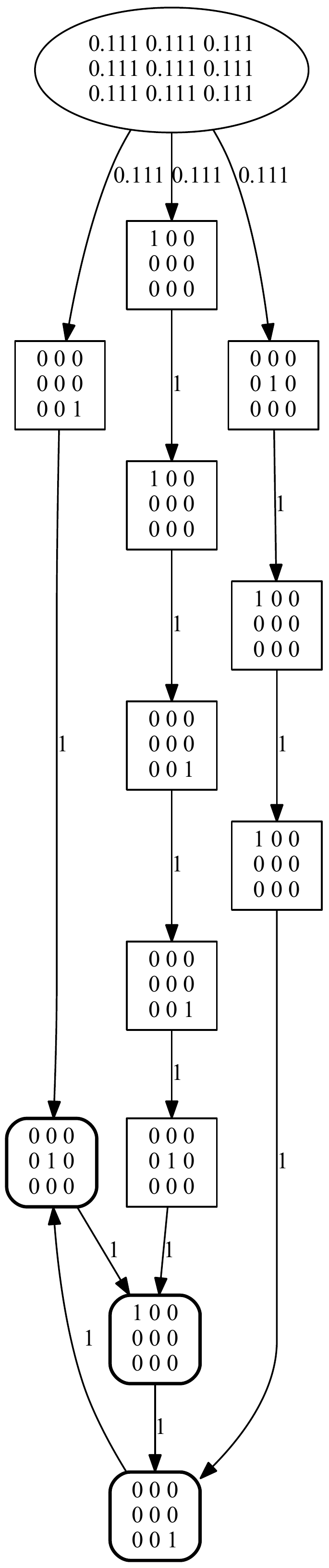}
\label{fig:shapleyfp}}
\hfil
\hfil
\subfigure[GFP and AFFFP]{
\includegraphics[scale=0.4]{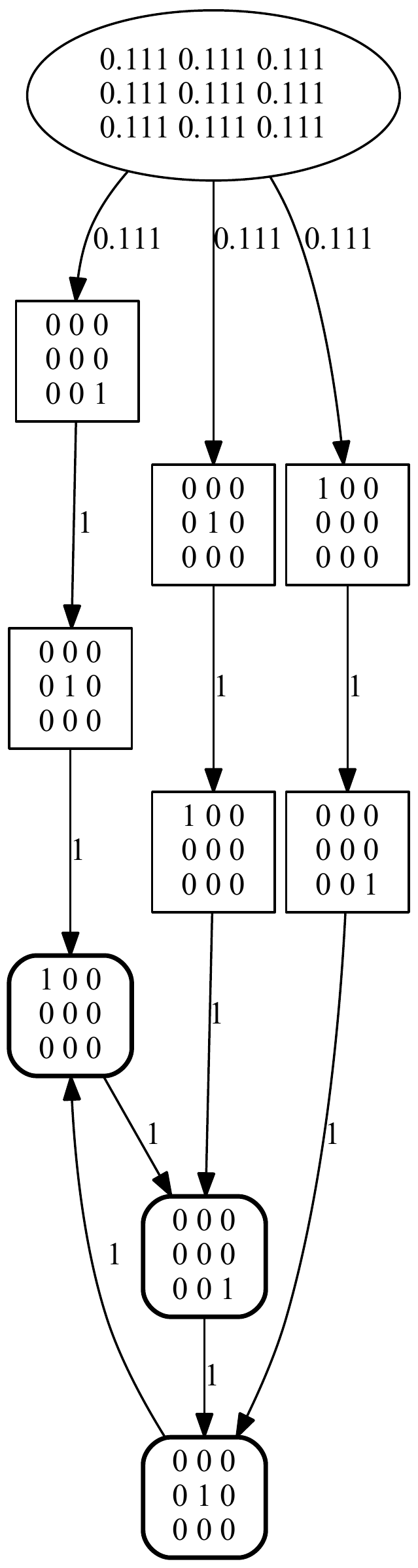}
\label{fig:shapleygfp}}
  \captionof{figure}{DTMC for the  mixed strategy equilibium in Shappley's game}   \label{fig:shapleydtmc}
\end{figure}

\subsection{Discussion}
%
%As we observe from Tables~\ref{table-results1} and \ref{table-results11},
%DSA is the least suitable algorithm in order to achieve coordination
%in the tested games. The nature of these games is one reason. Another
%is that DSA is a stochastic algorithm that will always deviate from a
%steady state even if its outcome is optimum. When we consider the
%number of states that were explored and the number of iterations that
%were reached in our tool, $RM$ and $MGRM$ generated
%significantly fewer number of states and run for fewer iterations
%than the other algorithms. We should mention also that 
The $20\times 20$ symmetric game a game, which was introduced in \cite{big_game} in
order to show the poor performance of $FP$. Nonetheless their reported results are not always correct since for the specific parameters FP converged almost always in the NE of the game. %Therefore it is possible for particular cases of interest that the performance of the learning algorithm   

% in some games, no other
%learning algorithms have been observed to  perform better than $FP$ under some
%specific parameters we adopted.

%% The differences in Figures \ref{fig:stag1} and \ref{fig:stag2} of
%% Section \ref{sec:sh}, indicate the importance of the initial
%% conditions in the outcome of the game. In particular, a specific level
%% of confidence is needed in order for players to choose the Pareto
%% efficient equilibrium of the Stag-Hunt game. When different initial
%% conditions of interest are tested using the proposed tool such 
%% differences can be revealed.

The experimental results are indicative of that no algorithm performed
perfectly in all games initialised.  
% As depicted in Figure \ref{fig:res}, 
The two variants of Regret Matching have better 
performance than the other algorithms on the simple coordination game,
but much worse performance on the complex coordination game than the FP
based algorithms. Although GFP has the best average performance in the
two coordination games, there exist initial conditions in the complex coordination
game where GFP fails to converge to Nash equilibrium. This suggests
that there could be no single learning algorithm that is optimal for
every competitive game or cooperative-game-based distributed
optimisation. Therefore, a selection method for learning algorithms,
such as our tool, is useful to have in practice.

\section{Conclusions} \label{sec:concl}

This paper presented a novel approach that can evaluate the
performance of game-theoretic learning algorithms over finite time. The
key complexity reduction technique is is the use of a {\em behaviour similarity relation}, which makes
verification feasible by reducing the number of states. The
effectiveness has been tested on a set of learning algorithms using
various examples. These examples  
illustrated the principles of our new method as  they have relatively well-known
dynamics and equilibra, which allowed us to carefully examine
the experimental results and confirm their correctness.
The cumulative performance result
clearly demonstrated the importance of using our approach in this
domain.
%
% Although only are used in the experiments, it was possible to
% demonstrate and compare players' behaviour.  

Future work will include investigating the application of the
techniques to multi-player games and cooperative games in robotics and
other real-world scenarios for distributed optimisation. Extending
our approach to more learning algorithms is another direction we want to
pursue.  The proposed behaviour similarity relation can be used for any myopic
algorithms and FP based algorithms which do not use randomisation for
computing $\sigma^{-i}$.  However, some FP based algorithms adopt
randomisation to predict other players' strategies, such as Extended
Kalman Filter Fictitious Play (EKFFP) \cite{ifac} and Particle Filter
Fictitious Play (PFFP)~\cite{cj}.  It is a challenge to define
behaviour similarity relation and calculate the upper bound of maximum
error for these algorithms. Another direction is to evaluate the
impact of initial strategies.

\section*{Acknowledgment}
This work was supported by the EPSRC project EP/J011894/2.

%% \bibliographystyle{IEEEtrans}
%% \bibliography{verification}

\begin{thebibliography}{10}
\providecommand{\url}[1]{#1}
\csname url@samestyle\endcsname
\providecommand{\newblock}{\relax}
\providecommand{\bibinfo}[2]{#2}
\providecommand{\BIBentrySTDinterwordspacing}{\spaceskip=0pt\relax}
\providecommand{\BIBentryALTinterwordstretchfactor}{4}
\providecommand{\BIBentryALTinterwordspacing}{\spaceskip=\fontdimen2\font plus
\BIBentryALTinterwordstretchfactor\fontdimen3\font minus
  \fontdimen4\font\relax}
\providecommand{\BIBforeignlanguage}[2]{{%
\expandafter\ifx\csname l@#1\endcsname\relax
\typeout{** WARNING: IEEEtranS.bst: No hyphenation pattern has been}%
\typeout{** loaded for the language `#1'. Using the pattern for}%
\typeout{** the default language instead.}%
\else
\language=\csname l@#1\endcsname
\fi
#2}}
\providecommand{\BIBdecl}{\relax}
\BIBdecl

\bibitem{smart_grid2}
T.~Ayken and J.-i. Imura, ``Asynchronous distributed optimization of smart
  grid,'' in \emph{Proc. SICE Annual Conference (SICE'12)}.\hskip 1em plus
  0.5em minus 0.4em\relax IEEE, 2012, pp. 2098--2102.

\bibitem{brown_fict}
G.~W. Brown, ``Iterative solutions of games by fictitious play,'' in
  \emph{Activity Analysis of Production and Allocation}, T.~C. Koopmans,
  Ed.\hskip 1em plus 0.5em minus 0.4em\relax Wiley, 1951, pp. 374--376.

\bibitem{emery2004}
R.~Emery-Montemerlo, G.~Gordon, J.~Schneider, and S.~Thrun, ``Approximate
  solutions for partially observable stochastic games with common payoffs,'' in
  \emph{Autonomous Agents and Multiagent Systems, 2004. AAMAS 2004. Proceedings
  of the Third International Joint Conference on}.\hskip 1em plus 0.5em minus
  0.4em\relax IEEE, 2004, pp. 136--143.

\bibitem{emery2005}
------, ``Game theoretic control for robot teams,'' in \emph{Proceedings of the
  2005 IEEE International Conference on Robotics and Automation}.\hskip 1em
  plus 0.5em minus 0.4em\relax IEEE, 2005, pp. 1163--1169.

\bibitem{learning_in_games}
D.~Fudenberg and D.~Levine, \emph{The theory of Learning in Games}, K.~Binmore,
  Ed.\hskip 1em plus 0.5em minus 0.4em\relax The MIT Press, 1998.

\bibitem{games1}
D.~Fudenberg and J.~Tirole, \emph{Game Theory}.\hskip 1em plus 0.5em minus
  0.4em\relax MIT Press, 1991.

\bibitem{stelios}
E.~Gelenbe and S.~Timotheou, ``Random neural networks with synchronized
  interactions,'' \emph{Neural Computation}, vol.~20, no.~9, pp. 2308--2324,
  2008.

\bibitem{big_game}
P.~W. Goldberg, R.~Savani, T.~B. S{\o}rensen, and C.~Ventre, ``On the
  approximation performance of fictitious play in finite games,'' \emph{Int. J.
  Game Theory}, vol.~42, no.~4, pp. 1059--1083, 2013.

\bibitem{monitoring}
J.~Kho, A.~Rogers, and N.~R. Jennings, ``Decentralized control of adaptive
  sampling in wireless sensor networks,'' \emph{ACM Transactions on Sensor
  Networks}, vol.~5, no.~3, p.~19, 2009.

\bibitem{sn}
------, ``Decentralized control of adaptive sampling in wireless sensor
  networks,'' \emph{ACM Trans. Sen. Netw.}, vol.~5, no.~3, pp. 1--35, 2009.

\bibitem{gwfp}
D.~S. Leslie and E.~Collins, ``Generalised weakened fictitious play,''
  \emph{Games and Economic Behavior}, vol.~56, no.~2, pp. 285 -- 298, 2006.

\bibitem{May}
B.~May, N.~Korda, A.~Lee, and D.~S. Leslie, ``Optimistic bayesian sampling in
  contextual-bandit problems,'' \emph{J. Mach. Learn. Res.}, vol.~13, pp.
  2069--2106, 2012.

\bibitem{nair}
R.~Nair, M.~Tambe, M.~Yokoo, D.~Pynadath, and S.~Marsella, ``Taming
  decentralized pomdps: Towards efficient policy computation for multiagent
  settings,'' in \emph{IJCAI}, 2003, pp. 705--711.

\bibitem{Nash}
J.~Nash, ``Equilibrium points in n-person games,'' in \emph{Proc. the National
  Academy of Science, USA}, vol.~36, 1950, pp. 48--49.

\bibitem{oliehoek}
F.~A. Oliehoek, M.~T. Spaan, J.~S. Dibangoye, and C.~Amato, ``Heuristic search
  for identical payoff bayesian games,'' in \emph{Proceedings of the 9th
  International Conference on Autonomous Agents and Multiagent Systems: volume
  1-Volume 1}.\hskip 1em plus 0.5em minus 0.4em\relax International Foundation
  for Autonomous Agents and Multiagent Systems, 2010, pp. 1115--1122.

\bibitem{KNP04a}
J.~Rutten, M.~Kwiatkowska, G.~Norman, and D.~Parker, \emph{Mathematical
  Techniques for Analyzing Concurrent and Probabilistic Systems}.\hskip 1em
  plus 0.5em minus 0.4em\relax Amer. Math. Soc., 2004.

\bibitem{shapley}
L.~Shapley, \emph{In advances in Game theory}.\hskip 1em plus 0.5em minus
  0.4em\relax Prinston University, 1964.

\bibitem{ifac}
M.~Smyrnakis and S.~Veres, ``Coordination of control in robot teams using
  game-theoretic learning,'' in \emph{Proc. IFAC'14}, 2014, pp. 1194--1202.

\bibitem{mythesis}
M.~Smyrnakis, ``Game-theoretical approaches to decentralised optimisation,''
  Ph.D. dissertation, University of Bristol, 2011.

\bibitem{cj}
M.~Smyrnakis and D.~S. Leslie, ``{Dynamic Opponent Modelling in Fictitious
  Play},'' \emph{The Computer Journal}, vol.~53, pp. 1344--1359, 2010.

\bibitem{sc}
A.~Stranjak, P.~S. Dutta, M.~Ebden, A.~Rogers, and P.~Vytelingum, ``A
  multi-agent simulation system for prediction and scheduling of aero engine
  overhaul,'' in \emph{Proc. AAMAS'08}, 2008, pp. 81--88.

\bibitem{tarjan72}
R.~Tarjan, ``Depth-first search and linear graph algorithms,'' \emph{SIAM
  Journal on Computing}, vol.~1, pp. 146--160, 1972.

\bibitem{smart_grid1}
T.~Voice, P.~Vytelingum, S.~D. Ramchurn, A.~Rogers, and N.~R. Jennings,
  ``Decentralised control of micro-storage in the smart grid.'' in \emph{AAAI},
  2011, pp. 1421--1426.

\end{thebibliography}
% Generated by IEEEtranS.bst, version: 1.13 (2008/09/30)

\newpage
\section*{Appendix: Reward matrix for the complex coordination game}
\label{sec:app2}
{\scriptsize
\begin{equation*} \label{big_reward}
r=\kbordermatrix{ &a_1 &a_2 &a_3  a_4  & a_5 & a_6 & a_7 & a_8 & a_8 & a_9 & a_{10}   \\
b_1 & 0 & 0 & 0 & 0 & 0 & 0 & 0 & 0 & 0 & 0   \\
b_2 & 1 & 0 & 0 & 0 & 0 & 0 & 0 & 0 & 0 & 0  \\
b_3 & 0.95594 & 1 & 0 & 0 & 0 & 0 & 0 & 0 & 0 & 0   \\
b_4 & 0.95594 & 0.95594 & 1 & 0 & 0 & 0 & 0 & 0 & 0 & 0   \\
b_5 & 0.95594 & 0.95594 & 0.95594 & 1 & 0 & 0 & 0 & 0 & 0 & 0   \\
b_6 & 0.95594 & 0.95594 & 0.95594 & 0.95594 & 1.2099 & 1 & 0 & 0 & 0 & 0   \\
b_7 & 0.95594 & 0.95594 & 0.95594 & 0.95594 & 0.95594 & 1.2099 & 1 & 0 & 0 & 0  \\
b_8 & 0.95594 & 0.95594 & 0.95594 & 0.95594 & 0.95594 & 0.95594 & 1.2099 & 1 & 0 & 0   \\
b_9 & 0.95594 & 0.95594 & 0.95594 & 0.95594 & 0.95594 & 0.95594 & 0.95594 & 1.2099 & 1 & 0   \\
b_{10} & 0.95594 & 0.95594 & 0.95594 & 0.95594 & 0.95594 & 0.95594 & 0.95594 & 0.95594 & 1.2099 & 1   \\
b_{11} & 0.95594 & 0.95594 & 0.95594 & 0.95594 & 0.95594 & 0.95594 & 0.95594 & 0.95594 & 0.95594 & 1.2099  \\
b_{12} & 0.95594 & 0.95594 & 0.95594 & 0.95594 & 0.95594 & 0.95594 & 0.95594 & 0.95594 & 0.95594 & 0.95594  \\
b_{13} & 0.95594 & 0.95594 & 0.95594 & 0.95594 & 0.95594 & 0.95594 & 0.95594 & 0.95594 & 0.95594 & 0.95594   \\
b_{14} & 0.95594 & 0.95594 & 0.95594 & 0.95594 & 0.95594 & 0.95594 & 0.95594 & 0.95594 & 0.95594 & 0.95594  \\
b_{15} & 0.95594 & 0.95594 & 0.95594 & 0.95594 & 0.95594 & 0.95594 & 0.95594 & 0.95594 & 0.95594 & 0.95594  \\
b_{16} & 0.95594 & 0.95594 & 0.95594 & 0.95594 & 0.95594 & 0.95594 & 0.95594 & 0.95594 & 0.95594 & 0.95594   \\
b_{17} & 0.95594 & 0.95594 & 0.95594 & 0.95594 & 0.95594 & 0.95594 & 0.95594 & 0.95594 & 0.95594 & 0.95594   \\
b_{18} & 0.95594 & 0.95594 & 0.95594 & 0.95594 & 0.95594 & 0.95594 & 0.95594 & 0.95594 & 0.95594 & 0.95594   \\
b_{19} & 0.95594 & 0.95594 & 0.95594 & 0.95594 & 0.95594 & 0.95594 & 0.95594 & 0.95594 & 0.95594 & 0.95594   \\
b_{20} & 0.95594 & 0.95594 & 0.95594 & 0.95594 & 0.95594 & 0.95594 & 0.95594 & 0.95594 & 0.95594 & 0.95594 \\ 
\cline{2-11}
& a_{11} & a_{12} & a_{13} & a_{14} & a_{15} & a_{16} & a_{17} & a_{18} & a_{19} & a_{20}  \\
b_1 & 0 & 0 & 0 & 0 & 0 & 0 & 0 & 0 & 0 & 0  \\
b_2 & 0 & 0 & 0 & 0 & 0 & 0 & 0 & 0 & 0 & 0  \\
b_3 & 0 & 0 & 0 & 0 & 0 & 0 & 0 & 0 & 0 & 0  \\
b_4 & 0 & 0 & 0 & 0 & 0 & 0 & 0 & 0 & 0 & 0  \\
b_5 & 0 & 0 & 0 & 0 & 0 & 0 & 0 & 0 & 0 & 0  \\
b_6 & 0 & 0 & 0 & 0 & 0 & 0 & 0 & 0 & 0 & 0  \\
b_7 & 0 & 0 & 0 & 0 & 0 & 0 & 0 & 0 & 0 & 0  \\
b_8 & 0 & 0 & 0 & 0 & 0 & 0 & 0 & 0 & 0 & 0  \\
b_9 & 0 & 0 & 0 & 0 & 0 & 0 & 0 & 0 & 0 & 0  \\
b_{10} & 0 & 0 & 0 & 0 & 0 & 0 & 0 & 0 & 0 & 0  \\
b_{11} & 1 & 0 & 0 & 0 & 0 & 0.95594 & 0.95594 & 0.95594 & 0.95594 & 1.2099  \\
b_{12} & 1.2099 & 1 & 0 & 0 & 0 & 0 & 0.95594 & 0.95594 & 0.95594 & 0.95594  \\
b_{13} & 0.95594 & 1.2099 & 1 & 0 & 0 & 0 & 0 & 0.95594 & 0.95594 & 0.95594  \\
b_{14} & 0.95594 & 0.95594 & 1.2099 & 1 & 0 & 0 & 0 & 0 & 0.95594 & 0.95594  \\
b_{15} & 0.95594 & 0.95594 & 0.95594 & 1.2099 & 1 & 0 & 0 & 0 & 0 & 0.95594  \\
b_{16} & 0.95594 & 0.95594 & 0.95594 & 0.95594 & 1.2099 & 1 & 0 & 0 & 0 & 0  \\
b_{17} & 0 & 0.95594 & 0.95594 & 0.95594 & 0.95594 & 1.2099 & 1 & 0 & 0 & 0  \\
b_{18} & 0 & 0 & 0.95594 & 0.95594 & 0.95594 & 0.95594 & 1.2099 & 1 & 0 & 0  \\
b_{19} & 0 & 0 & 0 & 0.95594 & 0.95594 & 0.95594 & 0.95594 & 1.2099 & 1 & 0  \\
b_{20} & 0 & 0 & 0 & 0 & 0.95594 & 0.95594 & 0.95594 & 0.95594 & 1.2099 & 1  
}
\end{equation*}
}

\end{document}